\documentclass[10pt]{article}
\usepackage[utf8]{inputenc}

\usepackage{fullpage}

\usepackage[linesnumbered,ruled]{algorithm2e}

\usepackage[dvipsnames]{xcolor}
\usepackage{todonotes}
\usepackage{xspace}
\usepackage{complexity}
\usepackage[linesnumbered,ruled]{algorithm2e}
\usepackage{tikz}
\usepackage{amsmath}
\usepackage{amsthm}
\usepackage{amssymb}
\usepackage{doi}
\usepackage{authblk}
\usepackage{subcaption}
\usepackage{enumitem}
\usepackage[capitalise]{cleveref}

\newtheorem{theorem}{Theorem}

\newtheorem{observation}{Observation}
\newtheorem{lemma}[theorem]{Lemma}

\newtheorem{claim}{Claim}

\newcommand{\YES}{\ensuremath{\mathsf{Yes}}\xspace}
\newcommand{\NO}{\ensuremath{\mathsf{No}}\xspace}
\newcommand{\TJ}{\textsc{TJ}\xspace}
\newcommand{\TS}{\textsc{TS}\xspace}

\newcommand{\TJA}{\ensuremath{\leftrightarrow_{\text{\textsc{TJ}}}}\xspace}
\newcommand{\TSA}{\ensuremath{\leftrightarrow_{\text{\textsc{TS}}}}\xspace}

\newcommand{\TSeq}{\ensuremath{\leftrightsquigarrow_{\text{\textsc{TS}}}}\xspace}

\newcommand{\TJR}{\textsc{Token Jumping}\xspace}
\newcommand{\TSR}{\textsc{Token Sliding}\xspace}

\newenvironment{claimproof}[1][\proofname]
  {%
    \proof[#1]%
  }
  {%
    \endproof%
  }

\title{Independent set reconfiguration in $H$-free graphs\footnote{The authors are supported by ANR project GrR (ANR-18-CE40-0032).}}

\author[1]{Valentin Bartier}
\author[1]{Nicolas Bousquet}
\author[2]{Moritz Mühlenthaler}
\affil[1]{Univ. Lyon, Université Lyon 1, CNRS, LIRIS UMR CNRS 5205, F-69621, Lyon, France}
\affil[2]{G-SCOP, Université Grenoble-Alpes, Grenoble, France}
\date{}

\begin{document}

\maketitle

\begin{abstract}
	Given a graph $G$ and two independent sets of $G$, the independent set reconfiguration problem asks whether one independent set can be transformed into the other by moving a single vertex at a time, such that at each intermediate step we have an independent set of $G$. We study the complexity of this problem for $H$-free graphs under the token sliding and token jumping rule. Our contribution is twofold. First, we prove a reconfiguration analogue of Alekseev's theorem, showing that the problem is \PSPACE-complete unless $H$ is a path or a subdivision of the claw. We then show that under the token sliding rule, the problem admits a polynomial-time algorithm if the input graph is fork-free.
\end{abstract}

\section{Introduction}

For some nominal combinatorial problem, such as the satisfiability problem or the independent set problem, a corresponding \emph{reconfiguration problem} asks the following question: Given two solutions to some instance of the nominal problem, can one be transformed into the other by a sequence of allowed modifications, such that each intermediate solution is feasible?
Such problems appear in many domains, such as combinatorial puzzles (Rubik's cube, Sokoban), motion planning, or in the context of phylogenetic trees, see~\cite{Heuvel:13,Nishimura:18} for an overview of recent results. The complexity of reconfiguration problems on graphs has received considerable attention recently, for example the reconfiguration of vertex covers \cite{hutchison_parameterized_2013, mouawadVC2014}, dominating sets \cite{HADDADAN201637, Suzuki2016, BONAMY20216}, colorings \cite{BOUSQUET20161,BONAMY202045,DVORAK2021103319,BDKLJ:23}, and independent sets~\cite{hutchison_reconfiguring_2014,Bartier2020OnGA}. The natural complexity class of reconfiguration problems is \PSPACE\xspace and we may ask under which conditions we can do better, that is, under which conditions a reconfiguration problem admits a polynomial-time algorithm. We consider this question for the reconfiguration of independent sets in graphs that are $H$-free, that is, graphs that do not contain a fixed graph $H$ as an induced subgraph.

Our study of independent set reconfiguration in $H$-free graphs is largely motivated by a similar line of research on the nominal problem of finding a maximum independent set (MIS) in $H$-free graphs. The input is an $H$-free graph $G$ and the task is to find an independent set of $G$ of maximum size. Despite roughly 40 years of effort, a complete classification of the complexity is not known. Alekseev showed that the problem remains \NP-hard for a large class of graphs $H$. 
\begin{theorem}[{\cite{alekseev1982effect}}]
	\label{thm:alekseev-h-free} 
	The maximum independent set problem on $H$-free graphs is \NP-hard, unless $H$ is a path, the claw, or a subdivision of the claw. 
\end{theorem}
Let $P_k$ denote the path graph on $k$ vertices and see Figure~\ref{fig:smallgraphs} for illustrations of $P_4$, the claw and the fork. Minty \cite{minty_maximal_1980} and Sbihi \cite{sbihi_algorithme_1980} showed independently that a maximum independent set of a claw-free graph can be computed in polynomial time. Twenty years later, Alekseev generalized this result to fork-free graphs \cite{alekseev_polynomial_2004}. Noticing that a maximum independent set can be computed in polynomial time $P_4$-free graphs, the smallest interesting case left open by Theorem \ref{thm:alekseev-h-free} is $H = P_5$. The complexity of this problem remained open for three decades, but recently, polynomial-time algorithms were obtained for $H = P_5$~\cite{lokshantov_independent_2014} and $H=P_6$~\cite{grzesik_polynomial-time_2020}. The complexity remains open for $H=P_\ell$ for $\ell \geq 7$.

\begin{figure}
    \begin{subfigure}{.3\textwidth}
    \centering
        \begin{tikzpicture}[vertex/.style={shape=circle,thick,draw,node distance=3em,inner sep=0.15em},edge/.style={draw,thick}]
            \node[vertex] (v1) {};
            \node[vertex,above of=v1] (v2) {};
            \node[vertex,above of=v2] (v3) {};
            \node[vertex,above of=v3] (v4) {};

            \draw[edge] (v1) -- (v2) -- (v3) -- (v4);
        \end{tikzpicture}
        \caption{$P_4$}
        \label{fig:p4}
    \end{subfigure}%
    \begin{subfigure}{.3\textwidth}
    \centering
        \begin{tikzpicture}[vertex/.style={shape=circle,thick,draw,node distance=3em,inner sep=0.15em},edge/.style={draw,thick}]
            \node[vertex] (v1) {};
            \node[vertex,above left of=v1] (v2) {};
            \node[vertex,below of=v1] (v3) {};
            \node[vertex,above right of=v1] (v4) {};

            \draw[edge] (v2) -- (v1) -- (v3);
            \draw[edge] (v1) -- (v4);
        \end{tikzpicture}
        \caption{claw}
    \label{fig:p4}
    \end{subfigure}%
    \begin{subfigure}{.3\textwidth}
    \centering
        \begin{tikzpicture}[vertex/.style={shape=circle,thick,draw,node distance=3em,inner sep=0.15em},edge/.style={draw,thick}]
            \node[vertex] (v1) {};
            \node[vertex,above left of=v1] (v2) {};
            \node[vertex,below of=v1] (v3) {};
            \node[vertex,above right of=v1] (v4) {};
            \node[vertex,below of=v3] (v5) {};

            \draw[edge] (v2) -- (v1) -- (v3);
            \draw[edge] (v1) -- (v4);
            \draw[edge] (v5) -- (v3);
        \end{tikzpicture}
        \caption{fork}
    \label{fig:p4}
    \end{subfigure}%
    \caption{Small graphs of interest.}
    \label{fig:smallgraphs}
\end{figure}
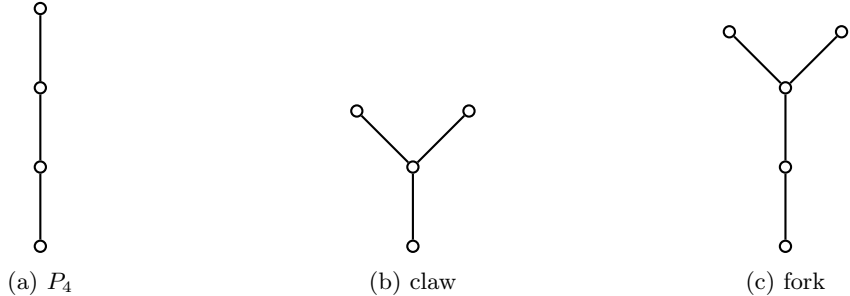

We follow the line of results above and investigate the complexity of \textsc{Independent Set Reconfiguration} (ISR) in $H$-free graphs. The ISR problem has been introduced in~\cite{ito_complexity_2011} and it has been studied widely since, see e.g.,~\cite{hutchison_reconfiguring_2014,  hutchison_parameterized_2014, ahn_fixed-parameter_2014, demaine_linear-time_2015, elbassioni_sliding_2015, bodlaender_token_2017, Bartier2020OnGA}. An instance of ISR is a graph $G$ and two independent sets $I$ and $J$ of $G$ of the same size, which are considered as \emph{tokens} placed on the vertices of $G$. The question is whether $I$ can be transformed into $J$ by a sequence of \emph{moves}, such that in each intermediate step the tokens form an independent set.
Two main types of moves have been considered: \emph{token jumping}, where a token can move to any other vertex of the graph (provided that it is not adjacent to a token), and \emph{token sliding}, where a token can slide along an edge to a neighboring vertex. We denote by \TJR and \TSR and the corresponding variants of ISR, where the only allowed moves are token jumping and token sliding, respectively. 
Both problems are equivalent for maximum independent sets and are known to be \PSPACE-complete, even for planar graphs of maximum degree three \cite{hearn_pspace-completeness_2005,ito_complexity_2011}.
Polynomial-time algorithms are known for example on trees \cite{demaine_linear-time_2015}, interval graphs \cite{bodlaender_token_2017}, bipartite permutation and bipartite distance hereditary graphs \cite{elbassioni_sliding_2015}. 

We prove the following reconfiguration analogue of Alekseev's theorem for ISR in $H$-free graphs.
\begin{theorem}
	\label{thm:alekseev-me}
	\TSR and \TJR in $H$-free graphs are $\PSPACE$-complete, unless $H$ is a path, the claw, or a subdivision of the claw.
\end{theorem}
It is known that \TSR and \TJR admit polynomial-time algorithm on $P_4$-free graphs~\cite{bonsma_independent_2016,kaminski_complexity_2012} and claw-free graphs~\cite{hutchison_reconfiguring_2014}. We generalize the results for \TSR, showing that the problem admits a polynomial-time algorithm in fork-free graphs. 
\begin{theorem}
	\label{thm:fork-free}
	\TSR in fork-free graphs admits a polynomial-time algorithm.
\end{theorem}
Let us note that the complexity of \TJR on fork-free graphs remains open. Furthermore, the smallest graph $H$ for which the complexity of \TSR in $H$-free graphs is open is $H=P_5$.

\paragraph{Overview of the proofs.}
In order to prove Theorem \ref{thm:alekseev-me}, we use similar ideas to those that Alekseev used to prove Theorem \ref{thm:alekseev-h-free}. Let $G_t$ be the graph obtained from a graph $G = (V, E)$ by subdividing each edge of $G$ exactly $t$ times. The main idea is to relate maximum independent sets of a graph $G$ with those of $G_t$. In \cite{alekseev1982effect} Alekseev showed that $\alpha(G_{2t}) = \alpha(G) + t|E|$, where
$\alpha(G)$ denotes the size of a
maximum independent set of $G$. We may therefore assume that the input graph has arbitrarily large girth and that vertices of degree at least three are arbitrarily far apart. Theorem~\ref{thm:alekseev-h-free} then follows from the fact that MIS is \NP-hard on graph of maximum degree three.
To prove Theorem~\ref{thm:alekseev-me}, we associate each maximum independent set of $G$ in a canonical way with a maximum independent set of $G_{2t}$.
Let $I$ and $J$ be two maximum independent sets of $G$ and let $I'$ and $J'$ be the corresponding canonical independent sets of $G_{2t}$. We show that there exists a transformation from $I$ to $J$ if and only if there exists a transformation from $I'$ to $J'$. Combining this with the fact that \TSR and \TJR are \PSPACE-complete on graphs of maximum degree three~\cite{hearn_pspace-completeness_2005,ito_complexity_2011} yields Theorem \ref{thm:alekseev-me}. 

We then show that \TSR in $H$-free admits a polynomial-time algorithm (Theorem~\ref{thm:fork-free}). Let $G$ be a fork-free graph and $I$ and $J$ be independent sets of $G$ of the same size.
We notice that a vertex $v$ of $G$ is \emph{irrelevant} for \TSR if no independent set reachable from $I$ contains $v$. Throughout the algorithm, we reduce the graph by deleting irrelevant vertices. For instance, any vertex with at least three tokens in its neighborhood is irrelevant. In general, we distinguish two cases, depending on whether $I$ and $J$ are maximum or not. 
In the first case, we provide a simple reduction to \TSR in claw-free graphs. We show that if $I$ is maximum then any induced claw of $G$ contains precisely one token on a degree-one vertex of the claw. This means that the center of any claw is irrelevant and may be removed from the graph. We thus obtain an equivalent instance on a claw-free graph. Notice that by the equivalence of \TSR and \TJR for maximum independent sets, this also yields a polynomial-time algorithm for \TJR in claw-free graphs if the independent sets are maximum.

The case that $I$ and $J$ are not maximum is more involved. We observe that, after removing irrelevant vertices, each connected component the graph induced by the vertices $I$ and $J$ is either a path or a cycle. Sliding tokens on induced paths is easy, so it remains to deal with the cycles and isolated vertices. Similar obstructions occur for alternating cycles in \textsc{Matching Reconfiguration}~\cite{ito_complexity_2011} and cycles induced by the symmetric difference of independent sets of claw-free graphs~\cite{hutchison_reconfiguring_2014}). We show that by sliding tokens along an $I$-augmenting path ($I$ is not maximum), we may assume that $G$ contains an $I$-free vertex, that is, a vertex $f$, such that $I + f$ is an independent set.
We then show that, for any vertex $c$ of a cycle in the graph induced by the symmetric difference of the two independent sets, either we can transform $I$ into $I - c + f$ or find an irrelevant vertex, which can be removed (Theorem~\ref{thm:ts_tj_equiv_fork}). In the former case, this allows us to move the remaining tokens on the cycle to their target positions in $J$. To prove Theorem~\ref{thm:ts_tj_equiv_fork}, we first show that we may assume that every connected component of $G$ is prime with respect to the modular decomposition. Notice that the modular decomposition was used ``in disguise'' in the form of a \emph{cotree} in the polynomial-time algorithm for ISR in cographs ($P_4$-free graphs) in~\cite{bonsma_independent_2016}. In order to move the token on $c$ to the free vertex $f$, we imagine sliding it on a shortest path $P$ from $c$ to $f$. Then $c$ may be blocked by other tokens on $P$ or in the neighborhood of $P$. Using a structural characterization of claws in fork-free prime graphs due to Brandst{\"a}dt et al.~\cite{brandstadt2004minimal}, we show we may move $c$ and every vertex that is blocking $c$ forward on the path, one by one in reverse order, or find an irrelevant vertex (see in particular Lemma~\ref{lem:rotation-claw}). 

\paragraph{Related work.} 

The complexity of independent set reconfiguration has been widely studied in the literature, for example in particular graph classes~\cite{demaine_linear-time_2015,bodlaender_token_2017,elbassioni_sliding_2015} and in the context of parameterized complexity~\cite{Nishimura:18,Ito:20}.
For $H$-free graphs we are aware of only two results: Bonsma~\cite{bonsma_independent_2016} showed that \TSR is solvable in polynomial on in $P_4$-free graphs. In~\cite{hutchison_reconfiguring_2014}, Bonsma et al.~showed that \TSR and \TJR admit a polynomial-time algorithm on claw-free graphs~\cite{hutchison_reconfiguring_2014}. In the case that more than one graph is forbidden (as induced subgraph), Kamiński et al.~\cite{kaminski_complexity_2012} showed that \TJR admits a polynomial-time algorithm on even-hole free graphs. Interestingly, the complexity of MIS is open for this graph class. 
The complexity of \TSR and \TJR can be very different, for example, Lokshtanov and Mouawad~\cite{lokshtanov_complexity_2019} showed that \TJR is \NP-complete on bipartite graphs (which are odd-hole free), whereas \TSR remains \PSPACE-complete. Similarly, \TJR admits a linear-time algorithm on split graphs, while \TSR is \PSPACE-complete on this graph class~\cite{belmonte_token_2021}. The reader is referred to~\cite{Bartier21} for a complete and recent bibliography on the topic. ISR has also been studied from the point of view of parameterized complexity, where the parameter is for example the size of the independent sets or the length of the reconfiguration sequence. The reader is referred to~\cite{ISRsurvey} for an overview of results in this direction.
Many other reconfiguration problems have been studied when an induced subgraph is forbidden, especially from a structural point of view. For single vertex recoloring for instance, a complete dichotomy theorem was recently obtained by Belavadi, Cameron, Merkel in~\cite{belavadi23}, who characterized exactly for which graphs $H$, $\G(G,\chi(G)+1)$ is connected for every $H$-free graph $G$ (where $\chi(G)$ is the chromatic number of $G$).

\paragraph{Organization.}

Section~\ref{sec:prelim} contains userful notation and definitions. In Section~\ref{sec:alekseev}, we prove the reconfiguration analogue of Alekseev's theorem (Theorem~\ref{thm:alekseev-me}). We discuss irrelevant vertices of fork-free graphs in \cref{sec:irrelevant}. We propose a polynomial-time algorithm for \TSR for maximum independent sets in fork-free graphs in Section~\ref{sec:max-isr} and for non-maximum independent sets in Section~\ref{sec:nonmax-ts}. Section~\ref{sec:conclusion} concludes the paper with a discussion of open questions.

\section{Preliminaries}
\label{sec:prelim}

For two sets $U, W$ we denote by $U \bigtriangleup W$ their symmetric
difference $(U \setminus W) \cup (W \setminus U)$. Graphs in this paper are
undirected and finite. Let $G = (V, E)$ be a graph. A set $I \subseteq V$ is
\emph{independent} in $G$ if the vertices in $I$ are pairwise
non-adjacent in $G$. We denote by $\alpha(G)$ the maximal size of an independent set of $G$. For a vertex $v \in V$ we denote by $N(v)$ the
\emph{neighborhood of $v$}, that is, the set of vertices  that are adjacent to
$v$ in $G$. For $t \geq 1$ we denote by $P_t$ the path graph on $t$ vertices. A
vertex of degree one of $P_t$ is called endvertex and a vertex of degree two of
$P_t$ is called \emph{internal vertex}. Given graph $G$, the graph $\overline{G}$ is the \emph{complement graph} of $G$ and is obtained by replacing every edge of $G$ by a non-edge en vice-versa. 

A \emph{claw} is a graph on four vertices that has one vertex of degree three (the \emph{center}) and  three vertices of degree one (the \emph{leaves}). A \emph{fork} is a graph obtained from a claw by subdividing one edge, that is, by replacing an edge by a path on three vertices. The \emph{center} of a fork $F$ is the  vertex of $F$ of degree $3$.
For $t \geq 1$, the \emph{$t$-subdivision} of $G$ is the graph obtained from
$G$ by replacing each edge $uv$ of $G$ with a copy of $P_{t+2}$ and identifying its endvertices with $u$
and $v$. We say that a $t$-subdivision of $G$ is \emph{even} if $t$ is even.

\paragraph{Modular decomposition.} 
A set $U \subseteq V$ of vertices is called
\emph{module} if for each vertex $v \in V \setminus U$, the vertex $v$ is either
adjacent to every vertex in $U$ or $v$ is non-adjacent to every vertex in $U$.
A module $U \subseteq V$ is \emph{non-trivial} if $U \neq \emptyset$ and $U
\neq V$.  By recursively dividing a connected graph into maximal modules we may compute in linear time its \emph{modular decomposition}~\cite{mcconnell1999modular}. The modular decomposition of a graph has been introduced by Gallai
in~\cite{gallai1967transitiv}. Using this decomposition, many graph problems that are hard in general can be solved efficiently on certain classes of graphs~\cite{mohring1985algorithmic,mohring1984substitution}. For instance, the modular decomposition is used in the polynomial-time
algorithm for computing a maximum-weight independent set in a fork-free
graph~\cite{lozin_polynomial_2008}. In the context of reconfiguration, the
modular decomposition of cographs (that are $P_4$-free graphs) has been used (implicitly) in order to obtain
a polynomial-time algorithm for \textsc{Token Jumping} on
cographs~\cite{bonsma_independent_2016} \footnote{Note that the algorithm is given in the TAR-model which is proven to be equivalent to the TJ-model for independent set reconfiguration.}.

\paragraph{Independent Set Reconfiguration.}
In the following we define two adjacency relations on the independent sets of
$G$ that correspond to three natural reconfiguration steps on independent
sets: \emph{token sliding} (\TS) and \emph{token jumping} (\TJ).
Let $I, J \subseteq V$ be two independent sets
of $G$. We say that $I$ and $J$ are \TJ-adjacent ($I \TJA J$) if there is
precisely one vertex in $I$ that is not in $J$ and vice versa. Furthermore, we
say that $I$ and $J$ are \TS-adjacent ($I \TSA J$) if $I \TJA J$ and the vertex
$u$ in $I \setminus J$ and the vertex $v$ in $J \setminus I$ are adjacent in
$G$. 
For each
rule $R \in \{\TS, \TJ, \}$ we say that $I$ and $J$ admit a
\emph{$R$-reconfiguration sequence} $I_0, I_1, \ldots, I_\ell$ and write $I \leftrightsquigarrow_R J$ if $I = I_0$, $J
= I_\ell$ and for $0 \leq i < \ell$ we have $I_i \leftrightarrow_{R} I_{i+1}$.
If it is clear from the context we may omit the rule $R$. Given a graph $G$ and
two independent sets $I, J \subseteq V(G)$, the problem \TSR (resp., \TJR) asks
whether $I$ and $J$ admit a \TS-reconfiguration sequence (resp.,
\TJ-reconfiguration sequence).

Note that \TSR and \TJR are equivalent on maximum independent sets since, if a vertex can jump on a non-adjacent vertex, the independent set is indeed not maximum.

\section{A reconfiguration analogue of Alekseev's theorem}
\label{sec:alekseev}
In this section we prove Theorem~\ref{thm:alekseev-me}, a reconfiguration
analogue of Alekseev's Theorem~\ref{thm:alekseev-h-free}).  In the
following, let $G = (V, E)$ be a graph and $G_t$ be the $t$-subdivision of $G$ for some even
$t \geq 1$. 
For any independent set $I$ of $G$, we can construct in a canonical way an independent set $I'$ of $G_t$ by placing for each edge $uv$ of $G$ on the internal vertices of the copy of
$P_{t+2}$ in $G_t$ that corresponds to $uv$ exactly $t/2$ additional tokens. The independent set $I'$ is called the \emph{extension of $I$}.
Alekseev proved the following:

\begin{theorem}[{\cite{alekseev1982effect}}] \label{thm:alex}
	The graph $G$ has a maximum independent set of size $k$ if and only if $G_t$
	has maximum independent set of size $k+ t\cdot |E(G)| / 2$. 
\end{theorem}

Theorem~\ref{thm:alex} implies immediately that for any maximum independent set $I$ of $G$, the extension $I'$ of $I$ in $G$ is also maximum.
Our goal in this section is to adapt Theorem~\ref{thm:alex} to \TSR and \TJR as follows.

\begin{theorem}
	\label{thm:alex-reconf}
    Let $I$ and $J$ be two maximum independent sets of $G$ and let $R \in \{\TS, \TJ\}$. Then $I \leftrightsquigarrow_R J$ with respect to $G$ if and only if $I' \leftrightsquigarrow_R J'$ with respect to $G_t$. 
\end{theorem}

The rest of this section is devoted to prove Theorem~\ref{thm:alex-reconf}.
Recall that \TSR and \TJR are $\PSPACE$-complete in graphs of maximum degree
three~\cite{ito_complexity_2011,hearn_pspace-completeness_2005}. By considering subdivisions of $G$, we may assume that $G$ has arbitrarily high girth. Therefore, just as Alekseev's Theorem \ref{thm:alex} implies Theorem
\ref{thm:alekseev-h-free}, it suffices to prove Theorem~\ref{thm:alex-reconf} to
obtain Theorem~\ref{thm:alekseev-me}.
The remainder of this section is devoted to the proof of
Theorem~\ref{thm:alex-reconf}.  Observe that by Theorem~\ref{thm:alex} the
independent sets $I'$ and $J'$ of Theorem~\ref{thm:alex-reconf} are maximum
independent sets of $G_t$. Recall that for maximum independent sets, \TSR and
\TJR are equivalent\footnote{If a token can jump to a non-neighbor, then the independent set is indeed not maximum.}, so it suffices to consider the \TS version.

For an edge $uv$ of $G$ let $S_{uv} = \{s_{uv}^1, s_{uv}^2,
\ldots, s_{uv}^t\}$ be the set of internal vertices of the copy of $P_{t+2}$ that replaces the edge $uv$ in the construction of
$G_t$. We assume that $s_1$ is adjacent to $u$ and $s_t$ is adjacent to $v$ and
that for $1 < i < t$, the vertex $s_i$ is adjacent to $s_{i+1}$ and $s_{i-1}$. 
A \emph{left-move} on $S_{uv}$ (assuming that we have an order between $u$ and $v$) is a move of a token on $s_i$ with $i > 1$ to $s_{i-1}$, provided that $s_{i-1}$ is not adjacent to a token.

\begin{claim}\label{cl:nr-subdiv-tokens}
    Let $uv$ be an edge of $G$ and let $I'$ be a maximum independent set of $G_t$. If $\{u, v\} \subseteq I'$ then $|I' \cap S_{uv}| = (t-2)/2$ and $|I' \cap S_{uv}| = t/2$ otherwise.
\end{claim}
\begin{proof} 
	Starting from $I'$, apply left-moves to $S_{uv}$ until no longer
	possible and let $I''$ be the resulting independent set.  Since no
	left-move can be applied to $S_{uv}$, if $\{u, v\} \subseteq I''$
	we have that for $1 \leq i < t$ even, there is a token of $I''$ on
	$s_i$. Since $t$ is even and $I'$ is maximum, there are precisely $(t-2)/2$ tokens in $S_{uv} \cap I''$. Since we only moved tokens on $S_{uv}$ to obtain $I''$, the same must hold for $I'$.
	On the other hand, suppose that $\{u, v\} \nsubseteq I''$ and (without
	loss of generality) that $u \notin I''$. Since no left-move can be
	applied to $S_{uv}$, any token of $I''$ is on some $s_i$ with $1 \leq i
	\leq t$ odd. Since $t$ is even and $I''$ is maximum, there must be a
	token on each $s_i$ for $1 \leq i < t$ odd and thus there are
	exactly $t/2$ tokens of $I''$ on $S_{uv}$, and the same holds for $I'$.
\end{proof}

Let $G_t$ be the $t$-extension of $G=(V,E)$. By abuse of notation, we will denote by $V$ the subset of vertices of $G_t$ corresponding to vertices of $G$.

Using Claim~\ref{cl:nr-subdiv-tokens}, we show that any two maximum independent sets of $G_t$ that agree on $V$ admit a \TS-reconfiguration sequence, and hence also a \TJ-reconfiguration sequence.

\begin{claim}
	\label{cl:same-trace-same-compo}
	Let $I'$ and $J'$ be two maximum independent sets of $G_t$ such that $I' \cap V = J' \cap V$. Then $I' \TSeq J'$.
\end{claim}
\begin{proof}
	Let $uv$ be an edge of $G$ and let $T := I' \cap V = J' \cap V$.
	If $\{u, v\} \subseteq T$ we have $I' \cap S_{uv} = J' \cap S_{uv} =
	\{s_{uv}^2, s_{uv}^4, \ldots, s_{uv}^{t-2}\}$ by Claim \ref{cl:nr-subdiv-tokens}. If
	$|\{u,v\} \cap T| = 1$, suppose without loss of generality that $u
	\notin T$. Then Claim \ref{cl:nr-subdiv-tokens} ensures that $I' \cap
	S_{uv} = J' \cap S_{uv} = \{s_{uv}^1, s_{uv}^{3} \ldots, s_{uv}^{t-1}\}$. Finally, if
	$\{u,v\} \cap T = \emptyset$, Claim \ref{cl:nr-subdiv-tokens} ensures
	that $|T| = t/2$. By applying any possible valid left-move on
	$S_{uv}$ for both $I'$ and $J'$ we obtain independent sets $I''$ and
	$J''$, respectively, such that $I'' \cap S_{uv} = J'' \cap S_{uv} =
	\{s_{uv}^1, s_{uv}^3, \ldots, s_{uv}^{t-1}\}$. Notice that $T = I'' \cap
	V = J'' \cap V$. It is then sufficient to repeat this process for
	any edge of $G$ that satisfies $\{u,v\} \cap T = \emptyset$ to obtain a
	\TS-reconfiguration sequence from $I'$ to $J'$. 
\end{proof}

The following claim, which ensures that a single reconfiguration step in $G$ can be adapted in $G_t$, will immediately imply the ``if''-direction of Theorem~\ref{thm:alex-reconf}. Indeed, it suffices to apply the following claim to each step of the transformation to get the conclusion.

\begin{claim}
    Let $I_1$, $I_2$ be two maximum independent sets of $G$ such that $I_1 \TSA I_2$ and let $I'_1$ and $I'_2$ be the extensions of $I_1$ and $I_2$ in $G_t$. Then $I'_1 \TSeq I'_2$.
\end{claim}
\begin{proof}
	Since $I_1 \TSA I_2$ let $e=uv$ be the edge such that $I_1 \bigtriangleup I_2 = \{u, v\}$ with $u \in I_1$ and $v \in I_2$.
    We show that the exists a reconfiguration sequence $I'_1 \TSeq I'_2$ in $G_t$. 
    
    First note that, except for $u$, there is no token on $w$ in $I_1'$ for every neighbor $w$ of $v$ distinct from $u$ (since $I_2 = I_1 + v - u$ is an independent set). So by Claim~\ref{cl:same-trace-same-compo}, $I'' = (I'_1 \setminus S_{vw}) \cup \{s_{vw}^2, s_{vw}^4, \ldots s_{vw}^t\}$ is a maximum independent set of $G_t$ that can be reached from $I_1'$. So, for every $w \in N(v) \setminus u$, we can assume that there is no token on $s_{vw}^1$. Let us denote by $I_1''$ the resulting independent set. We can slide in $I_1''$ the token $s_{uv}^t$ of $I''_1$ to $v$.
    Finally, it suffices to slide the token from $s_{uv}^i$ to $s_{uv}^{i+1}$ for every $i$ such that $s_{uv}^i \in I'_1$ (in decreasing order) and then slide the token from $u$ to $s_{uv}^1$ in order to obtain $I'_2$ from $I'_1$. We conclude that $I'_1 \TSeq I'_2$.
\end{proof}

Let us now prove the converse direction of Theorem~\ref{thm:alex-reconf}. We start with the following lemma:

\begin{claim}
    \label{cl:3-path-trace}
    Let $I'$ be a maximum independent set of $G_t$. Then there are no three vertices $u, v, w \in I' \cap V$, such that $uv, vw \in E(G)$.
\end{claim}
\begin{proof}
    Assume for a contradiction that there are three vertices $u, v, w \in I' \cap V$ such that $uv, vw \in E(G)$.
    By Claim~\ref{cl:nr-subdiv-tokens} and~\ref{cl:same-trace-same-compo}, we can transform $I'$ into an independent $J$ such that $J \cap S_{uv} = \{s_{uv}^2, s_{uv}^4, \ldots, s_{uv}^{t-2}\}$ and $J \cap S_{vw} = \{ s_{vw}^3, s_{vw}^5, \ldots, s_{vw}^{t-1}\}$. 
    But then we can slide the token on $v$ to $s_{uv}^t$ and obtain an independent $I''$ such that $s_{vw}^{1}$ has no token in its neighborhood in $G_t$, which contradicts the maximality of $I'$. 
\end{proof}

A graph $G'$ is an induced subgraph of $G$ if there exists a subset of vertices $V' \subseteq V$ such that the graph with vertex $V'$ where $uv$ is an edge if and only if $uv$ is an edge of $G$.
Let $I'$ be a maximum independent set of $G_t$. By the previous claim, the graph $G[I' \cap V]$ consists of isolated vertices and edges. Let $v(I')$ (resp, $e(I')$) denote the number of isolated vertices (resp., isolated edges) of $G[I' \cap V]$.
\begin{claim}\label{cl:size-trace}
    $v(I') + e(I') = \alpha(G)$.
\end{claim}
\begin{proof}
    Notice that since $I'$ is a maximum independent set of $G_t$ it has size $\alpha(G) + t\cdot |E(G)| / 2$ by Theorem~\ref{thm:alex}.
    Let $uv$ be an edge of $G$. By Claim \ref{cl:nr-subdiv-tokens}, we have $|S_{uv} \cap I'| = \frac{t-2}{2}$ if $\{u,v\} \subseteq I'$ 
    and $|S_{uv} \cap I'| = \frac{t}{2}$ otherwise. It follows that
    \begin{align*}
        \alpha(G) + t\cdot |E(G)| / 2 &= |I'| \\
             &= v(I') + \left(\sum_{\substack{uv \in E(G) \, :\\  |\{ u, v \} \cap I' | \leq 1}} \frac{t}{2} \right) + 2e(I') + \left( \sum_{\substack{uv \in E(G) \\ \,:\, u,v \in I'}} \frac{t-2}{2} \right) \\
             &= v(I') + e(I') + t\cdot |E(G)| / 2
    \end{align*}
    from which we conclude that $v(I') + e(I') = \alpha(G)$.
\end{proof}

To complete the proof of Theorem~\ref{thm:alex}, consider a \TS-reconfiguration sequence $I'_1 \TSA I'_2 \TSA \dots \TSA I'_\ell$ of maximum independent sets of $G_t$ such that $I' = I'_1$ and $J' = I_\ell$, where $I'$ (resp., $J'$) is the extension of $I$ (resp. $J$) in $G_t$. Notice that $G[I' \cap V]$ and $G[J' \cap V]$ contain no edge. And by Claim~\ref{cl:size-trace}, we can associate to each independent set $I'$ an induced subgraph consisting of $\alpha(G)$ isolated vertices or edges. 
We obtain from $I'$ a maximum independent set $I_i$ of $G$ by picking all isolated vertices and exactly one endpoint of each edge of $G[I'_i \cap V]$, say the first vertex with respect to an arbitrary fixed ordering of $V$. 
To conclude the proof it is enough to show that, for every $i \in \{1, 2, \ldots, \ell-1\}$, we have $I_i \TSeq I_{i+1}$. 
\begin{claim}
    \label{cl:aleks:G_t}
        $|(I'_i \cap V) \bigtriangleup (I'_{i+1} \cap V)| \leq 1 \enspace.$
\end{claim}
\begin{proof}
    Since $I'_i \TSA I'_{i+1}$ there exists an edge $xy$ such that $I'_{i+1}= (I_i \setminus x) \cup y$.  Note that at most one vertex in $x,y \in V$ since no edge of $G_t$ has both endpoints in $V$. We distinguish three cases:
    \begin{enumerate}
        \item $x, y \notin V$. Then $\{x, y\} \subseteq S_{vw}$ for some edge $vw$ of $G$ and therefore $|(I'_i \cap V) \bigtriangleup (I'_{i+1} \cap V)| = 0 \leq 1$.
        \item $x \in V$. Then $y \in S_{vw}$ for some edge $vw$ of $G$ and therefore $|(I'_i \cap V) \bigtriangleup (I'_{i+1} \cap V| = 1$.
        \item $y \in V$. Then $x \in S_{vw}$ for some edge $vw$ of $G$ and therefore $|(I'_i \cap V) \bigtriangleup (I'_{i+1} \cap V| = 1$.
    \end{enumerate}
\end{proof}

By claims~\ref{cl:size-trace} and~\ref{cl:aleks:G_t}, we just have to distinguish the following cases to complete the proof of the converse direction:
\begin{enumerate}
    \item $e(I'_i) = e(I'_{i+1})$. In this case we have $I_i = I_{i+1}$ and we are done.
    \item $e(I'_i) = e(I'_{i+1})+1$. Let $uv$ be the edge of $G$ that is not present in $G[I'_{i+1} \cap V]$ and assume without loss of generality that $u \in I_{i+1}$. It follows that $I'_{i+1}$ is obtained from $I'_i$ by sliding a token from $v$ to $S_{uv}^p$. Then either $u \in I_{i}$, so we have $I_i = I_{i+1}$ or $v \in I_{i}$ and we have $I_i \TSA I_{i+1}$.
    \item $e(I'_i) = e(I'_{i+1})-1$. The same argument can be applied by exchanging $I'_i$ and $I'_{i+1}$.
\end{enumerate}
So, for every $1 \leq i < \ell$, we have $I_i \TSeq I_{i+1}$. And then $I \TSeq J$ which completes the proof.

\section{Irrelevant vertices in fork free graphs}\label{sec:claw-three}
\label{sec:irrelevant}

In the following, let $G = (V, E)$ be a fork-free graph and $I$ and $J$ be two independent sets of $G$. We will identify vertices of $G$ that are irrelevant in any transformation from $I$ to $J$ in the sense that they cannot contain a token. We show that vertices adjacent to at least two tokens are irrelevant if they satisfy a certain neighborhood-condition. As a consequence, all vertices adjacent to at least three tokens are useless and can be safely deleted.
We say that a set $X \subseteq V$ of vertices is \emph{locally blocked} (with respect to $I$) if any vertex in $X$ is adjacent to at least two tokens in $I$. 
We are interested in locally blocked vertices that remain locally blocked under any token move. Let $X \subseteq V$ be a locally blocked subset of vertices. We say that the set $B_X := N(X) \cap I$ is \emph{$X$-blocking}: no token can slide from $B_X$ to $X$, since there is another token in $B_X$ blocking it. 

Let $u \in V$ and $X \subseteq V - u$. A vertex $v$ of $G - (X + u)$ is an \emph{$X$-twin of $u$} if $uv$ is an edge and $N(v) \cap X = N(u) \cap X$.
A set $X \subseteq V$ of vertices \emph{permanently blocked} if it is locally blocked (with respect to $I$) and for each independent set $I'$ such that $I \TSeq I$ and for each vertex $u$ in the blocking set $B_X = N(X) \cap I'$, any $I'$-free neighbor  $v$ of $u$ is an $X$-twin of $u$. (It means that if we move, in $I'$ a token in the neighborhood of $X$ from $u$ to another vertex $v$ then the neighborhood in $u$ and $v$ in $X$ are the same).
We first introduce the following reduction rule, which states that permanently blocked vertices can be safely removed from the graph.

\begin{description}\label{cl:Z}
    \item[Reduction Rule Z] Let $X \subseteq V$ be permanently blocked. Return $(G-X, I, J)$ if $X \cap J = \emptyset$ and a trivial \NO-instance otherwise.
\end{description}

A reduction rule is \emph{safe} if applying it yields an equivalent instance.

\begin{lemma}\label{lem:ruleZ}
    Reduction Rule Z is safe.
\end{lemma}
\begin{proof}
    First notice that for any independent set $I'$ such that $I \TSeq I'$, the set $X$ is locally blocked with respect to $I'$, since whenever a token $u \in B_X$ moves to an $I'$-free neighbor $v$ then $v$ is an $X$-neighbor of $u$. In particular, no independent set reachable from $I$ contains a token in $X$.

    If $J \cap X \neq \emptyset$ then $(G, I, J)$ is a \NO-instance, since for any independent set reachable from $I$, the set $X$ is locally blocked. 
    Suppose that $J \cap X = \emptyset$ and that $I \TSeq J$ with respect to $G$. Since no independent set $I'$ reachable from $I$ in $G$ contains a token in $X$, we have $I \TSeq J$ with respect to $G$ if and only if $I \TSeq J$ with respect to $G - X$.
\end{proof}

At this point it is not quite clear how to use Reduction Rule Z algorithmically. We will later prove that certain subsets of vertices are permanently blocked (and can be found in polynomial time). For now, we show that any vertex of $G$ adjacent to at least three tokens of $I$ is permanently blocked and can thus be safely deleted according to Lemma~\ref{lem:ruleZ}. 

\begin{lemma}\label{lem:ktoken}
    Let $c \in V$ be a vertex that is adjacent to exactly $k \geq 3$ tokens in $I$ and let $I'$ be an independent set such that $I' \TSeq I$. Then $c$ is adjacent to exactly $k$ tokens of $I'$.
\end{lemma}
\begin{proof}
    Let $J'$ be an independent set such that $I \TSeq J'$.
    Suppose for a contradiction that $c$ is adjacent to $k' \neq k$ tokens of $J'$. In a reconfiguration sequence from $I$ to $J'$, let $I'$ be the first independent set such that $c$ is adjacent to a number of tokens of $I'$ different from $k$. Suppose first that $|N(c) \cap I'| < k$. Let $x \in N(c)$ be the token that moved to some neighbor $x'$ to obtain $I'$. Since $N(c)$ contains $k \geq 3$ tokens in any independent set preceding $I'$ in the sequence, we have that $I'$ contains at least two tokens $y, z \in N(c)$ different from $x$. Since $x', y, z \in I'$, we have that $x'$ must be adjacent to $c$, otherwise $\{x', x, c, y, z\}$ induces a fork. But then $|N(c) \cap I'| = k$, a contradiction. Now suppose that $|N(c) \cap I'| > k$. Then some token $x' \notin N(c)$ moved to a vertex $x \in N(c)$. Let $y, z \in N(c) \cap I'$ be two tokens different from $x$. But then $\{x', x, c,y,z\}$ induces a fork, a contradiction. 
\end{proof}

Lemma~\ref{lem:ktoken} immediately implies the following sufficient condition for a vertex to be permanently blocked.
\begin{observation}\label{lem:ruleA}
    Let $c \in V$ be a vertex that is adjacent to at least three tokens in $I$. Then $c$ is permanently blocked.
\end{observation}

From the safeness of Reduction Rule Z and the previous observation we conclude that the following Reduction Rule is safe.
\begin{description}
    \item[Reduction Rule A.] Let $c \in V$ such that $c$ is adjacent to at least three tokens in $I$. Return $(G-c, I, J)$ if $c \notin J$, otherwise return a trivial \NO-instance.\label{rule:A}
\end{description}

We say that $G$ is $I$-reduced if Reduction Rule A is not applicable with respect to $G$ and $I$. That is, $G$ contains no vertex that is adjacent to at least three $I$-tokens. Notice that by applying Reduction Rule A exhaustively, we may obtain from $G$ in polynomial time a fork-free graph $G'$ that is $I$-reduced and $J$-reduced. The resulting equivalent instance is either $(G', I, J)$ or a trivial \NO-instance. The following observations are consequences  of Lemma~\ref{lem:ktoken}.

\begin{observation}\label{lem:only-reducible}
    If $G$ is $I$-reduced then $G$ is $I'$-reduced for every independent set $I'$ such that $I \TSeq I'$.
\end{observation}

\begin{observation}\label{lem:at-most-two}
    Let $G$ be $I$-reduced and let $I'$ be and independent set such that $I \TSeq I'$. Then for any vertex $v$ of $G$ we have $|N(v) \cap I'| \leq 2$. 
\end{observation}

\section{Maximum independent set reconfiguration in fork-free graphs}
\label{sec:max-isr}

As a warm-up, let us prove the following theorem.
\begin{theorem} \label{thm:misr-fork-me}
    \TSR restricted to maximum independent sets admits a polynomial-time algorithm.
\end{theorem}
To do so, we provide a simple reduction of Maximum Independent Set Reconfiguration in fork-free graphs to the same problem in claw-free graphs, which can be decided in polynomial time \cite{hutchison_reconfiguring_2014}. In the following, let $G = (V, E)$ be a fork-free graph and $I$ and $J$ be two maximum independent sets of $G$. Note that it suffices to consider the token sliding rule only since $I$ and $J$ are maximum.
By \cref{lem:ruleZ} and \cref{lem:ruleA} we may assume that $G$ is both $I$-reduced and $J$-reduced. Therefore, $G$ contains no induced claw that contains three tokens of $I$. Since $I$ is maximum, the following stronger property holds.

\begin{lemma}
    \label{cl:one-leaf-in-I}
    Let $\{c, x, y, z\}$ induce a claw in $G$ with center $c$. Then $|\{x, y, z\} \cap I| = 1$.
\end{lemma}
\begin{proof}
    Since $I$ is maximum and $G$ is $I$-reduced, we have $1 \leq |\{c, x, y, z\} \cap I| \leq 2$. 
    Assume for a contradiction that $|\{c, x, y, z\} \cap I| = 2$ and let without loss of generality $\{c, x, y, z\} \cap I = \{x, y\}$.

    Since $I$ is maximum, the vertex $z$ must have a neighbor $t$ in $I$ and since $c$ has at most two neighbors in $I$, we have $t \notin N(c)$. It follows that $\{c,x,y,z,t\}$ induces a fork in $G$, contradiction.

    So it remains to show that $c \notin I$. Assume for a contradiction that $c \in I$. Since $I$ is maximum, there are at least three tokens in the neighborhood of $\{x, y, z\}$: otherwise we add $\{x, y, z\}$ to $I$ and delete $N(\{x, y, z\})$ from $I$ to obtain a larger independent set. Let $u, v \in I$ be distinct from $c$ in the neighborhood of $\{x,y,z\}$. If $u$ (resp., $v$) is adjacent to only one vertex in $\{x, y, z\}$ then $\{c, x, y, z, u\}$ (resp., $\{c, x, y, z, v\}$) induces a fork. Therefore, the vertices $u$ and $v$ have a common neighbor in $\{x, y, z\}$, say $x$. But then $x$ has three neighbors in $I$, which contradicts our assumption that $G$ is $I$-reduced. 
\end{proof}

By the previous lemma, no maximum independent set of $G$ has a token on the center of any induced claw of $G$. Therefore, the following reduction rule is safe.

\begin{description}
    \item[Reduction Rule MIS.] Let $\{c,x,y,z\}$ induce a claw in $G$ with center $c$. Return $(G-c, I, J)$.
\end{description}

\begin{proof}[Proof of~\cref{thm:misr-fork-me}]
The induced claws of $G$ can be enumerated in polynomial time (by brute-force), so Reduction Rule MIS can be applied exhaustively in polynomial time. We obtain an equivalent instance $(G', I, J)$ of \TSR such that the graph $G'$ is claw-free.
It remains to apply the algorithm from~\cite{hutchison_reconfiguring_2014} on $(G', I, J)$ to decide the initial instance $(G, I, J)$. Notice that any reconfiguration sequence from $I$ to $J$ in $G'$ is a reconfiguration sequence from $I$ to $J$ in $G'$.    
\end{proof}

\section{Non-maximum Token Sliding in fork-free graphs}
\label{sec:nonmax-ts}

We provide a polynomial-time algorithm for \TSR on fork-free graphs, generalizing the algorithm for claw-free graphs given in~\cite{hutchison_reconfiguring_2014}. In the following, let $G = (V, E)$ be a fork-free graph and let $I, J \subseteq V$ be two independent sets of $G$ of the same size. If $I$ and $J$ are maximum, we may apply Theorem~\ref{thm:misr-fork-me} and are done. Otherwise, we show that we may assume that each connected component of $G$ is prime (\cref{sec:prime}). The algorithm then proceeds as follows. We consider the symmetric difference of the two independent sets and notice that each connected component of $G[I \bigtriangleup J]$ is a path or an (even) cycle, each of which is can be treated separately. Paths are easy to deal with while cycles require more work. In order to deal with the cycles we use an idea that has been used for \textsc{Matching Reconfiguration} in~\cite{ito_complexity_2011}: we push tokens along a path from the cycle to a vertex that is $I$-free (that is, it has no neighbor in $I$), which allows us to move the remaining tokens on the cycle to their target positions. Showing that pushing the tokens along such a path is possible is non-trivial and is discussed in Section~\ref{sec:ifree}. It turns out that if we get stuck at any point then there is a permanently blocked vertex that can be removed from the graph. The complete algorithm is presented in Section~\ref{sec:resolve-cycles}.

\subsection{Reduction to prime graphs}\label{sec:reduction-to-prime}
\label{sec:prime}

We show that we may assume without loss of generality that the input graph $G$ is prime with respect to the modular decomposition. First, we observe that if a module of a graph contains at least two tokens then the instance induced by the vertices of this module can be treated independently since no token will be able to leave the module at any step.

\begin{observation}\label{obs:separate-module}
    Let $M$ be a module of $G$ such that $\ell:=|M \cap I|$ and $I'$ be such that $I \TSeq I'$. If $\ell \ge 2$, we have  $|M \cap I'| = \ell$ and if $\ell \le 1$ we have $|M \cap I'| \le 1$.
\end{observation}

By~\cref{obs:separate-module}, if a non-trivial module $M$ contains at least two tokens of $I$ and $J$, respectively, then all the vertices in the neighborhood of $M$ can be deleted. Otherwise, if $G[M]$ is connected and contains at most one token of $I$ and $J$, respectively, then all vertices of $M$ are equivalent with respect to token sliding, so the vertices of $M$ can be identified. However, the situation is slightly more complicated if $M$ is not connected since a token might have to leave the module to reach its target position. In that case we show that either a token in $M$ can move to a different connected component of $G[M]$ in exactly two moves, or it can never leave its connected component. 

\begin{lemma}\label{lem:move-cc-modules}
    Assume that $G$ is $I$-reduced, let $M$ be a non-trivial module of $G$, and let $C_1, C_2, \ldots, C_\ell$ be the connected components of $G[M]$. Suppose that $\ell \geq 2$ and $M \cap I = \{v\} \subseteq V(C_1)$. Then exactly one of the following holds:
    \begin{itemize}
        \item The vertex $v$ has a neighbor $c \in V \setminus M$ such that $N(c) \cap I = \{v\}$.
        \item $N(M) \setminus M$ is permanently blocked. In particular, we have $|I' \cap V(C_1)|= 1$ for every independent set $I'$ of $G$ such that $I \TSeq I'$.
    \end{itemize}
\end{lemma}
\begin{proof}
If $v$ has a neighbor $c \in V(G) \setminus M$ such that $N(c) \cap I = \{v\}$ then we can slide the token on $v$ to $c$ to obtain an independent set $I' \TSA I$ that satisfies $|I' \cap C_1| = 0$. 

If $v$ has no such neighbor, let us prove that $N(M) \setminus M$ is permanently blocked. Since $G$ is $I$-reduced and by the assumptions on $v$, each vertex $u$ of $N(M)$ has exactly two neighbors in $I$ and in any independent set reachable from $I$ contains at most two neighbors of $u$. Assume for a contradiction that there exists a vertex $u \in N(M)$ and an independent set $I'$ such that such that $I \TSeq I'$ and $|N(u) \cap I'| = 1$. Chose $I'$ such that the reconfiguration sequence from $I$ to $I'$ is shortest and let $J'$ be the independent set before $I'$ in the sequence. By Observation~\ref{obs:separate-module}, we have $|J' \cap N(u)| = 1$. Let $z \in (N(u)\cap J')\setminus M$. To obtain $I'$ from $J'$, the token on $z$ slides to $z' \notin N(u)$. Let $a, b \in N(u) \cap M$ such that $a$ and $b$ are in different components of $M$. But then $\{u,a,b,z,z'\}$ induces a fork, a contradiction.
\end{proof}

Suppose that $M$ is a module of $G$ such that that $M\cap I$ and $M\cap J$ have size at most one. We denote by $(G, I, J)_{/M}$ the instance of \TSR where $M$ has been replaced by a single vertex $m$ and the independent set $I$ (resp., $J$) contains $m$ if $I \cap M$ (resp., $J \cap M$) has size one\footnote{We still denote by $I$ and $J$ the independent sets by abuse of notations}.
We say that an instance of $(G,I,J)$ of \TSR is \emph{balanced} if for every non-trivial module $M$ of $G$ satisfying $|M \cap I| = |M \cap J| = 1$, the unique vertex in $M \cap I$ and the unique vertex in $M \cap J$ are in the same connected component of $G[M]$. 
By applying the next reduction rule exhaustively, we obtain either a trivial \NO-instance or a balanced instance.

\begin{description}
    \item[Reduction Rule B.] Assume there is a non-trivial module $M$ of $G$ such that $M \cap I = \{u\}$, $M \cap J = \{v\}$ and $u$ and $v$ are in two distinct components of $G[M]$. If there is a vertex $c \in V \setminus M$ such that $N(c) \cap I = \{u\}$, return $(G, I, J)_{/M}$. Otherwise, return a trivial \NO-instance. \label{rule:B}
\end{description}

\begin{lemma}\label{obs:contracte-module}
    Reduction Rule B is safe.
\end{lemma}
\begin{proof}
    Let $M$ be non-trivial module of $G$ such that $M \cap I = \{u\}$, $M \cap J = \{v\}$ and $u$ and $v$ are in two distinct components of $G[M]$.
    If there is a vertex $c \in V \setminus M$ such that $N(c) \cap I = \{u\}$, we may move the token on $u$ to $c$ and then to $v$, to obtain an independent set $I' \TSeq I$. Since any vertex of $M$ has the same neighborhood in $V \setminus M$, the module $M$ can contain at most one token in any independent set reachable from $I$. Since $v \in I' \cap J$, we may safely identify all vertices in $M$.
    
    On the other hand, assume that there is no vertex $c \in V \setminus M$ such that $N(c) \cap I = \{u\}$. Let $C$ be the vertex set of the connected component of $G[M]$ containing $u$. By Lemma~\ref{lem:move-cc-modules}, for any independent set $I'$ such that $I' \TSeq I$, we have $|I' \cap C| = 1$, but $|J \cap C| = 0$. Therefore $(G, I, J)$ is a \NO-instance.
\end{proof}

Assume that the instance $(G, I, J)$ of \TSR is balanced. Then Observation \ref{obs:separate-module} tells us that if the input graph contains a non-trivial module the graph can be either reduced or split into connected components that can be treated independently. We now formalize this observation by introducing further reduction rules that are then applied exhaustively under the assumption that Reduction Rule B is not applicable. Each connected component of the resulting graph will be prime. 

\begin{description}
    \item[Reduction Rule D] Let $M$ be a non-trivial module of $G$ such that $|M \cap I| \leq 1$. If    $|M \cap J| \leq 1$ then return $(G, I, J)_{/M}$, otherwise return a trivial \NO-instance.
    \item[Reduction Rule E] Let $M$ be a non-trivial module of $G$ such that $|M \cap I| \ge 2$.  If $|M \cap J| \ne |M \cap I|$ return a trivial \NO-instance, otherwise return $(G - N(M), I, J)$.
\end{description}

Observe that if $G$ is prime then none of the reduction rules B, D, E are applicable. 

\begin{lemma}\label{lem:split-proc-safe}
    Suppose that Reduction Rule B is not applicable. Then reduction rules D and E are safe.
\end{lemma}
\begin{proof}
    Suppose first that Reduction Rule D is applicable and let $\{u \} = M \cap I$. By Observation~\ref{obs:separate-module} and Lemma~\ref{lem:move-cc-modules}, if $|M \cap J| > 1$ then $(G, I, J)$ is a \NO-instance; so assume that $\{v\} = M \cap J$. Given that Reduction Rule B is not applicable, the tokens on $u$ and $v$ are in the same connected component of $M$, so we may slide the token on $u$ to $v$. Notice that in any independent set reachable from $I$ and $J$, the module $M$ contains at most one token. Therefore, we may contract the module $M$ to a single vertex to obtain the instance $(G, I, J)_{/M}$ which is equivalent to $(I, J, M)$.

Now suppose that Reduction Rule E is applicable. By Observation~\ref{obs:separate-module} and Lemma~\ref{lem:move-cc-modules}, if $I \cap M$ and $J\cap M$ do not have the same size then $I$ and $J$ do not admit a \TS-reconfiguration sequence. Otherwise, by Lemma~\ref{lem:move-cc-modules}, no token of $M$ can leave the module, so we can treat $G[M]$ and $G[V \setminus N[M]]$ independently.
\end{proof}

We now prove we may consider without loss of generality instances of \TSR on fork-free graphs that are prime.
\begin{theorem}\label{thm:reduction-to-prime}
    Let $(G,I,J)$ be an instance of \TSR such that $G$ is fork-free. We may compute in polynomial from $(G, I, J)$ an equivalent instance $(G', I', J')$ of \TSR on a fork-free graph $G'$ that is $I'$-reduced and $J'$-reduced, such that each connected component of $G'$ is prime.
\end{theorem}
\begin{proof}
  We apply Reductions Rules A,B,D,E as long as we can (in this order) to obtain in polynomial time the instance $(G', I', J')$ of \TSR. Assume for contradiction that there is some connected component $C$ of $G'$ contains a non-trivial module $M$. If $M$ contains more than one token then Reduction Rule E can be applied. Otherwise, Reduction Rule D can be applied, a contradiction. Furthermore, $(G', I', J')$ and $(G, I, J)$ are equivalent by the safeness of the reduction rules (see \cref{lem:ruleA,lem:split-proc-safe,obs:contracte-module}). Finally, $G'$ is fork-free, since fork-free graphs are closed under taking induced subgraphs and contracting modules.
\end{proof}

\subsection{Reaching I-free vertices}\label{sec:equiv-ts-tj}
\label{sec:ifree}
In the following, let $G = (V, E)$ be a fork-free graph and let $I$ be an independent set of $G$ such that $G$ is $I$-reduced (Reduction Rule A is not applicable). Our goal in this section is to prove the following theorem. 

\begin{theorem}\label{thm:ts_tj_equiv_fork}
    Let $v \in I$ and let $u$ be an $I$-free vertex
    of $G$. Then $J := I - v + u$ is an independent set of $G$. 
    If $G$ is prime then there is a polynomial-time algorithm that computes a reconfiguration sequence $I\TSeq J$ or finds a permanently blocked vertex of $G$.
\end{theorem}

First, we provide the intuition behind the proof. Consider a shortest path $P$ from $u$ to $v$ in $G$: Either we can directly slide step by step a token along the path from $v$ to $u$, or some other tokens $a_1, a_2, \ldots, a_t \in I$ are in the (closed) neighborhood of the path $P$ and therefore prevent $v$ from moving to $u$. Let $a_1, a_2, \ldots, a_t$ be ordered by increasing distance to $v$. The idea is to move $a_t$ to $v$ and then for $i$ for $t-1$ down to $1$,  move $a_i$ to $a_{i+1}$. Finally, $u$ may move to $a_1$. The resulting independent set is precisely $J = I - v + u$, so $I \TSeq J$.
However, two tokens in $\{a_1, \ldots, a_t\}$ might block each other (if both have the same neighbors on the path). We first prove that if $P$ has length at least three then this is impossible. We use a characterization of claws in prime graphs due to Brandst{\"a}dt et al. \cite{brandstadt2004minimal} in order to handle the case where $P$ has length two. In this case we may not be able to directly reach $J$, but then some vertex of $G$ is permanently blocked and can therefore be deleted safely by Lemma~\ref{lem:ruleZ}.

Let $u, v \in V$ and let $P = u_1, \ldots u_\ell$ be a $u-v$ path in $G$. The \emph{leftmost} (resp. \emph{rightmost}) neighbor of $x \in N(V(P))$ on $P$ is the vertex of $V(P) \cap N(x)$ with the lowest (resp. highest) index. 

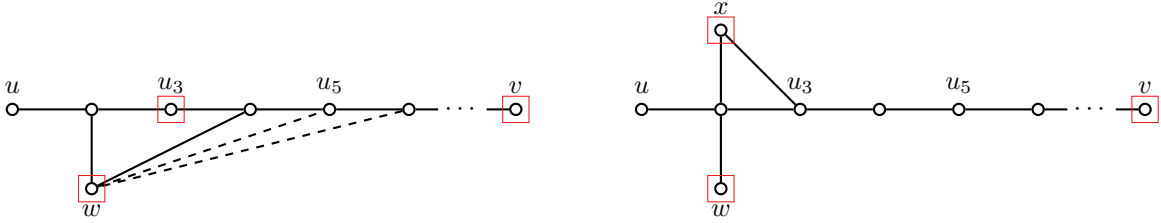
\begin{figure}
    \centering
    \begin{subfigure}[t]{.45\linewidth}
        \begin{tikzpicture}[vertex/.style={shape=circle,thick,draw,node distance=3em,inner sep=0.15em},edge/.style={draw,thick}]
            \node[vertex,label=above:$u$] (u1) {};
            \node[vertex,right of=u1] (u2) {};
            \node[vertex,right of=u2,label=above:$u_3$] (u3) {};
            \node[vertex,right of=u3] (u4) {};
            \node[vertex,right of=u4,label=above:$u_5$] (u5) {};
            \node[vertex,right of=u5] (u6) {};
            \node[right=0.75em of u6] (dots) {$\cdots$};
            \node[vertex,right=0.75em of dots,label=above:$v$] (u7) {};
            \node[vertex,below of=u2,label=below:$w$] (w) {};

            \draw[edge] (u1) -- (u2) -- (u3) -- (u4) -- (u5) -- (u6);
            \draw[edge] (u2) -- (w) -- (u4);
            \draw[edge,dashed] (u6) -- (w) --(u5);

            \draw [edge] (u6.0) -- +(0:0.3);
            \draw [edge] (u7.180) -- +(180:0.3);

            \path (u7.center) coordinate[draw,rectangle,minimum size=3.5mm,inner sep=0pt,outer sep=0pt,Red];
            \path (u3.center) coordinate[draw,rectangle,minimum size=3.5mm,inner sep=0pt,outer sep=0pt,Red];
            \path (w.center) coordinate[draw,rectangle,minimum size=3.5mm,inner sep=0pt,outer sep=0pt,Red];
        \end{tikzpicture}
    \end{subfigure}
    \hspace{2em}
    \begin{subfigure}[t]{.45\linewidth}
        \begin{tikzpicture}[vertex/.style={shape=circle,thick,draw,node distance=3em,inner sep=0.15em},edge/.style={draw,thick}]
            \node[vertex,label=above:$u$] (u1) {};
            \node[vertex,right of=u1] (u2) {};
            \node[vertex,right of=u2,label=above:$u_3$] (u3) {};
            \node[vertex,right of=u3] (u4) {};
            \node[vertex,right of=u4,label=above:$u_5$] (u5) {};
            \node[vertex,right of=u5] (u6) {};
            \node[right=.75em of u6] (dots) {$\cdots$};
            \node[vertex,right=0.75em of dots,label=above:$v$] (u7) {};
    
            \node[vertex,below of=u2,label=below:$w$] (w) {};
            \node[vertex,above of=u2,label=above:$x$] (x) {};
            
            \draw[edge] (u1) -- (u2) -- (u3) -- (u4) -- (u5) -- (u6);
            \draw[edge] (u3) -- (x) -- (u2) -- (w);
            \draw [edge] (u6.0) -- +(0:0.3);
            \draw [edge] (u7.180) -- +(180:0.3);
            
            \path (u7.center) coordinate[draw,rectangle,minimum size=3.5mm,inner sep=0pt,outer sep=0pt,Red];
            \path (x.center) coordinate[draw,rectangle,minimum size=3.5mm,inner sep=0pt,outer sep=0pt,Red];
            \path (w.center) coordinate[draw,rectangle,minimum size=3.5mm,inner sep=0pt,outer sep=0pt,Red];
        \end{tikzpicture}
        \end{subfigure}
    \caption{An illustration in the two cases of the proof of Lemma \ref{lem:first-neighbor-path}. The dashed lines on the left indicate non-edges and the red squares indicate tokens.}
    \label{fig:equiv}
\end{figure}

\begin{lemma}\label{lem:first-neighbor-path}
	Let $u$ be an $I$-free vertex of $G$ and let $v \in I$. Let $P$ be
	a shortest path from $u$ to $v$ in $G$. If $P$ has length at least three then
	the neighbor of $u$ in $P$ has at most one neighbor in $I$.
\end{lemma}
\begin{proof}
	Let $P = u_1, \ldots, u_\ell$ be a shortest $u - v$ path in $G$ such that $\ell \geq 4$. Suppose for a contradiction that $u_2$ has two neighbors in $I$ (it cannot have more than two since $G$ is $I$-reduced). There are two cases to consider. 

	\paragraph{Case 1:} $u_3 \in I$ and $u_2$ has a neighbor $w \in N(V(P))$ in $I$ (see Figure \ref{fig:equiv} (left)).  
	Notice that $u$, $u_3$, and $w$ are pairwise non-adjacent since $u$ is $I$-free and $u_3, w \in I$. Since $\ell \geq 4$, the edge $wu_4$ exists, otherwise $\{u_2, u_1, u_3, w, u_4\}$ induces a fork. Since $v, w \in I$,  we have $u_4 \neq v$, so $\ell \geq 5$. If $u_5 = v$, then $u_4$ has three neighbors in $I$, a contradiction to our assumption that $G$ is $I$-reduced, so $\ell \geq 6$. Then either $wu_5$ or $wu_6$ is an edge of $G$, so $P$ cannot be a shortest $u - v$ path, a contradiction.

	\paragraph{Case 2:}  $u_2$ has two neighbors $w,x \in I$ that are not on $P$ (see Figure \ref{fig:equiv} (right)). Let $u_i, u_j$ be the rightmost neighbors of $w$ and $x$ on $P$ respectively. First assume that $i = j$. Then $u_j \neq v$ since $v \in I$ and $u_{j+1} \neq v$ since  $G$ is $I$-reduced, so $\ell \geq i+2$. But then $\{u_i, w, x, u_{i+1}, u_{i+2}\}$ induces a fork, a contradiction. We therefore assume without loss of generality that $j > i \geq 2$. If $j > 3$ then $\{u_2, u_1, w, x, u_j\}$ induces a fork. If $j = 3$ then $w$ has no other neighbor on $P$ but $u_2$ (since $w$ is non-adjacent to $u$ and $u_i$). Then we can slide the token on $x$ to $u_3$: if not, then another token on a vertex $z$ must prevent this move. But then $\{u_2, u_1, w, u_3, z\}$ induces a fork. We may therefore slide the token on $x$ to $u_3$ and the previous case applies. 
\end{proof}

\begin{lemma}\label{lem:no-two-neigh-path}
Let $u$ be an $I$-free vertex of $G$, $v \in I$, and let $P$ a be shortest path $u-v$ path in $G$. If $P$ has length at least three then no two vertices of $N[V(P)] \cap I$ can have the same leftmost neighbor on $P$. 
\end{lemma}
\begin{proof} 
	Let $P = u_1, u_2, \dots, u_\ell$, such that $\ell \geq 4$.
	Suppose for a contradiction that there exists $w,x \in N[V(P)] \cap I$
	that have the same leftmost neighbor $u_i \in P$.
	By Lemma \ref{lem:first-neighbor-path}, we
	have $i > 2$. Notice that $u_{i-2}$ and $u_i$ are non-adjacent. Since $u_i$ is the leftmost neighbor of $w$ and $x$, we have that $\{w, x, u_i, u_{i-1}, u_{i-2}  \}$ induces a fork.
\end{proof}

The two previous lemmas allow us to deal with the case where a shortest path
from $v \in I$ to $u \notin N[I]$ has length at least three. The following
theorem, due to Brandst{\"a}dt et al., will be useful in order to deal with the
case where a shortest $u - v$ path has length two:

\begin{figure}
    \centering
    \begin{subfigure}[t]{.2\linewidth}
        \centering
        \begin{tikzpicture}[vertex/.style={shape=circle,thick,draw,node distance=3em,inner sep=0.15em},edge/.style={draw,thick}]
            \node[vertex,label=left:$c$] (c) {};
            \node[vertex,above of=c,label=left:$u$] (u) {};
            \node[vertex,below of=c,label=left:$w$] (w) {};
            \node[vertex,right of=c,label=right:$v$] (v) {};
            \node[vertex,above of=v,label=right:$x$] (x) {};
            \node[vertex,below of=v,label=right:$y$] (y) {};
            \draw[edge] (u) -- (x) -- (v) -- (y) -- (w);
            \draw[edge,ultra thick] (w) -- (c) -- (u);
            \draw[edge,ultra thick] (c) -- (v);
        \end{tikzpicture}
        \caption{$H_1$\label{fig:h1}}
    \end{subfigure}
    \begin{subfigure}[t]{.2\linewidth}
        \centering
        \begin{tikzpicture}[vertex/.style={shape=circle,thick,draw,node distance=3em,inner sep=0.15em},edge/.style={draw,thick}]
            \node[vertex,label=left:$c$] (c) {};
            \node[vertex,above of=c,label=left:$u$] (u) {};
            \node[vertex,below of=c,label=left:$w$] (w) {};
            \node[vertex,right of=c,label=right:$v$] (v) {};
            \node[vertex,above of=v,label=right:$x$] (x) {};
            \node[vertex,below of=v,label=right:$y$] (y) {};
            \draw[edge] (u) -- (x) -- (v) -- (y) -- (w);
            \draw[edge,ultra thick] (w) -- (c) -- (u);
            \draw[edge,ultra thick] (c) -- (v);
            \draw[edge] (x) to [bend left=55] (y);
        \end{tikzpicture}
        \caption{$H_2$\label{fig:h2}}
    \end{subfigure}
    \begin{subfigure}[t]{.2\linewidth}
        \centering
        \begin{tikzpicture}[vertex/.style={shape=circle,thick,draw,node distance=3em,inner sep=0.15em},edge/.style={draw,thick}]
            \node[vertex,label=left:$c$] (c) {};
            \node[vertex,above of=c,label=left:$u$] (u) {};
            \node[vertex,below of=c,label=left:$w$] (w) {};
            \node[vertex,right of=c,label=right:$v$] (v) {};
            \node[vertex,above of=v,label=right:$x$] (x) {};
            \node[vertex,below of=v,label=right:$y$] (y) {};
            \draw[edge] (u) -- (x) -- (v) -- (y) -- (w);
            \draw[edge,ultra thick] (w) -- (c) -- (u);
            \draw[edge,ultra thick] (c) -- (v);
            \draw[edge] (c) -- (y);
        \end{tikzpicture}
        \caption{$H_3$\label{fig:h3}}
    \end{subfigure}
    \begin{subfigure}[t]{.2\linewidth}
        \centering
        \begin{tikzpicture}[vertex/.style={shape=circle,thick,draw,node distance=3em,inner sep=0.15em},edge/.style={draw,thick}]
            \node[vertex,label=left:$c$] (c) {};
            \node[vertex,above of=c,label=left:$u$] (u) {};
            \node[vertex,below of=c,label=left:$w$] (w) {};
            \node[vertex,right of=c,label=right:$v$] (v) {};
            \node[vertex,above of=v,label=right:$x$] (x) {};
            \node[vertex,below of=v,label=right:$y$] (y) {};
            \draw[edge] (u) -- (x) -- (v) -- (y) -- (w);
            \draw[edge,ultra thick] (w) -- (c) -- (u);
            \draw[edge,ultra thick] (c) -- (v);
            \draw[edge] (x) to [bend left=55] (y);
            \draw[edge] (c) -- (y);
        \end{tikzpicture}
        \caption{$H_4$\label{fig:h4}}
    \end{subfigure}
    \begin{subfigure}[t]{.2\linewidth}
        \centering
        \begin{tikzpicture}[vertex/.style={shape=circle,thick,draw,node distance=3em,inner sep=0.15em},edge/.style={draw,thick}]
            \node[vertex,label=left:$c$] (c) {};
            \node[vertex,above of=c,label=left:$u$] (u) {};
            \node[vertex,below of=c,label=left:$w$] (w) {};
            \node[vertex,right of=c,label=right:$v$] (v) {};
            \node[vertex,above of=v,label=right:$x$] (x) {};
            \node[vertex,below of=v,label=right:$y$] (y) {};
            \node[vertex,right of=v,label=right:$z$] (z) {};
            \draw[edge] (u) -- (x) -- (v) -- (y) -- (w);
            \draw[edge,ultra thick] (w) -- (c) -- (u);
            \draw[edge,ultra thick] (c) -- (v);
            \draw[edge] (x) to [bend left=45] (y);
            \draw[edge] (x) -- (c) -- (y);
            \draw[edge] (x) -- (z) -- (y);
        \end{tikzpicture}
        \caption{$H_5$\label{fig:h5}}
    \end{subfigure}
    \begin{subfigure}[t]{.2\linewidth}
        \centering
        \begin{tikzpicture}[vertex/.style={shape=circle,thick,draw,node distance=3em,inner sep=0.15em},edge/.style={draw,thick}]
            \node[vertex,label=left:$c$] (c) {};
            \node[vertex,above of=c,label=left:$u$] (u) {};
            \node[vertex,below of=c,label=left:$w$] (w) {};
            \node[vertex,right of=c,label=right:$v$] (v) {};
            \node[vertex,above of=v,label=right:$x$] (x) {};
            \node[vertex,below of=v,label=right:$y$] (y) {};
            \node[vertex,right of=v,label=right:$z$] (z) {};
            \draw[edge,ultra thick] (u) -- (x) -- (v);
            \draw[edge] (v) -- (y) -- (w);
            \draw[edge,ultra thick] (x) -- (z);
            \draw[edge] (w) -- (c) -- (u);
            \draw[edge] (c) -- (v);
            \draw[edge] (x) to [bend left=45] (y);
            \draw[edge] (x) -- (c) -- (y);
            \draw[edge] (z) -- (y);
        \end{tikzpicture}
        \caption{$H_5$ (2nd claw)\label{fig:h52}}
    \end{subfigure}
    \begin{subfigure}[t]{.2\linewidth}
        \centering
        \begin{tikzpicture}[vertex/.style={shape=circle,thick,draw,node distance=3em,inner sep=0.15em},edge/.style={draw,thick}]
            \node[vertex,label=left:$c$] (c) {};
            \node[vertex,above of=c,label=left:$u$] (u) {};
            \node[vertex,below of=c,label=left:$w$] (w) {};
            \node[vertex,right of=c,label=right:$v$] (v) {};
            \node[vertex,above of=v,label=right:$x$] (x) {};
            \node[vertex,below of=v,label=right:$y$] (y) {};
            \node[vertex,right of=v,label=right:$z$] (z) {};
            \draw[edge] (u) -- (x) -- (v);
            \draw[edge,ultra thick] (v) -- (y) -- (w);
            \draw[edge] (x) -- (z);
            \draw[edge] (w) -- (c) -- (u);
            \draw[edge] (c) -- (v);
            \draw[edge] (x) to [bend left=45] (y);
            \draw[edge] (x) -- (c) -- (y);
            \draw[edge,ultra thick] (z) -- (y);
        \end{tikzpicture}
        \caption{$H_5$ (3rd claw)\label{fig:h53}}
    \end{subfigure}
    \caption{The fork-free prime expansions of a claw. The claws considered in Lemma \ref{lem:rotation-claw} are marked by thick edges. The graph $H_5$ is shown three times, once for each of its induced claws.}
    \label{fig:claw_expansions}
\end{figure}
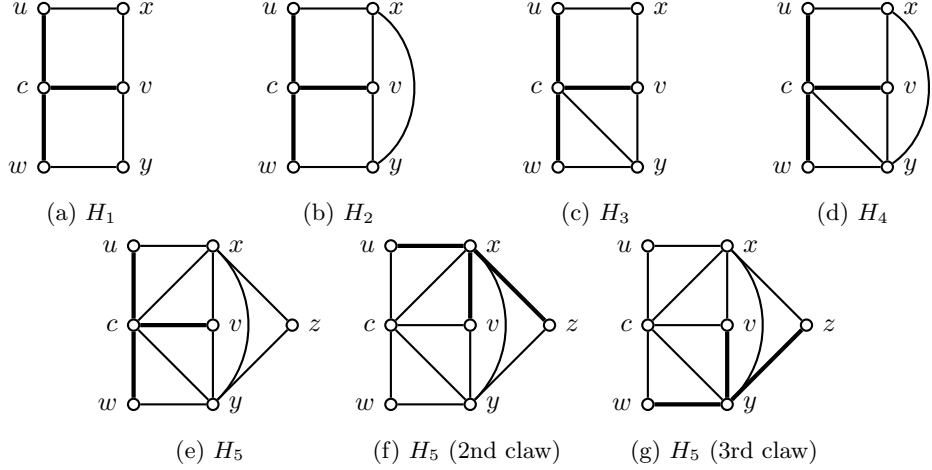

\begin{theorem}[\cite{brandstadt2004minimal}]\label{thm:extension-claw}
    Let $G$ be a fork-free prime graph which contains a claw induced by the vertices $\{c, u, v ,w\}$. Then one of the graphs $H_1, \ldots, H_5$ shown in Figure \ref{fig:claw_expansions} is an induced subgraph of $G$.
\end{theorem}

We use this characterization of the prime extensions of a claw in a fork-free graph to prove that if we are given a claw containing two tokens, we can either ``rotate'' the tokens around the center of the claw locally without moving any other token or that the center of the claw is permanently blocked and can thus be safely deleted from the graph by Lemma~\ref{lem:ruleZ}.

\begin{lemma}\label{lem:rotation-claw}
Let $I$ be an independent set of a fork-free prime graph $G$ such that $G$ is $I$-reduced. Let $\{c,u,v,w\}$ induce a claw in $G$ with center $c$ such that $\{u,v\} \subseteq I$. Then exactly one of the following holds:
\begin{enumerate}
    \item There exists a reconfiguration sequence from $I$ to $I - v + w$ and from $I$ to $I - u + v$ such that only the tokens initially on $u$ and $v$ are moved and such a sequence can be found in polynomial time, or
    \item the vertex $c$ is permanently blocked.
\end{enumerate}
\end{lemma}

\begin{proof}
    We first consider the case where the induced subgraph of $G$ containing the claw $\{c, u, v, w\}$ is $H_5$ shown in Figure~\ref{fig:claw_expansions}. Notice that the graphs $H_1, \ldots, H_4$ are subgraphs of $H_5$. In the following we will use two arguments: First, we show that certain edges must exist since otherwise there is an induced fork. Second, we show that certain edges cannot be present since Reduction Rule A is not applicable. In almost all cases, if these arguments apply to $H_5$ they also apply to $H_1, \ldots, H_4$, so it suffices to treat $H_5$. Finally, there is an automorphism between the first, the second and the third claw in $H_5$ shown in Figure~\ref{fig:claw_expansions}. That is, the second and the third claw can be obtained from the first by renaming vertices (preserving adjacencies and non-adjacencies), so exactly the same arguments apply as for the first claw of $H_5$.
    Let $I$ be an independent set of $G$ containing the vertices $u$ and $v$. The lemma follows from the following claims.

    \begin{claim}
        We can move the token $v$ of $I$ to $w$ via $y$.
    \end{claim}
    \begin{claimproof}
        We show that neither $w$ nor $y$ have a token in their neighborhood that is preventing us from moving the token on $v$ to $w$. Suppose that $y$ has a neighbor $y'$ containing a token. Then $c$ and $x$ are non-adjacent to $y'$ since Reduction Rule A is not applicable. Furthermore, $w$ and $y'$ are non-adjacent since otherwise $\{y'w'c'u'v\}$ induces a fork. But then $\{w,y,y',x,u\}$ induces a fork if $xy$ is present ($H_2$, $H_4$, $H_5$) and $\{w,y,y',v,x\}$ induces a fork otherwise ($H_1$ and $H_3$). Therefore $y'$ does not exist. Now suppose that $w$ has a neighbor $w'$ that contains a token. 
        Then $c$ and $w'$ are non-adjacent since Reduction Rule A is not applicable. But then $\{w',w,c,u,v\}$ induces a fork, so the token $w'$ does not exist. Therefore, we may move the token on $v$ to $w$ via $y$.
    \end{claimproof}

    \begin{claim} \label{claim:w-to-v}
        Either we can move the token $w$ of $I-v+w$ to $v$ via $y$ or the vertex $c$ is permanently blocked.
    \end{claim}
    \begin{claimproof}
        We try to move the token on $w$ to $v$ via $y$.
        The vertex $v$ has no neighbor $v'$ containing a token if two tokens are on $u$ and $w$, since otherwise $\{u,c,w,v,v'\}$ induces a fork.
        Therefore, we can move $w$ to $v$ via $y$ unless $y$ has a neighbor $y'$ with a token on it. In this case $c$ is non-adjacent to $y'$ since Reduction Rule A is not applicable. If $x$ is adjacent to $y$ ($H_2$, $H_4$, $H_5$) then $x$ is also adjacent to $y'$, since otherwise$\{u,x,y,w,y'\}$ induces a fork. Otherwise, if $x$ is non-adjacent to $y$ ($H_1$ and $H_3$) then $x$ is adjacent to $y'$, otherwise $\{x,v,y,w,y'\}$ induces a fork. It follows that $X = \{c, x, y\}$ is locally blocked, i.e., each vertex of $X$ has precisely two tokens in $B_X = \{u, w, y'\}$ in its neighborhood. In order to apply Lemma~\ref{lem:ruleZ}, we show that if we slide any token $d \in B_X$ to a vertex $d'$ then $d'$ is an $X$-twin of $d$.

        \begin{itemize}
            \item Suppose that $u \in B_X$ slides to a neighbor $u'$.  Then $y$ and $u'$ are non-adjacent by Observation~\ref{lem:only-reducible} and $v$ and $u'$ are non-adjacent since otherwise $\{w,y,z,v,u'\}$ induces a fork. Furthermore, $c$ and $u'$ are adjacent since otherwise $\{v, c,w,u,u'\}$ induces a fork. Finally, $x$ and $u'$ are adjacent since otherwise $\{u',u,x,v,y'\}$ induces a fork. Therefore, the vertex $u'$ is an $X$-neighbor of $u$.
            \item Suppose that $w \in B_X$ slides to a neighbor $w'$. Then $x$ and $w'$ are non-adjacent (otherwise $\{u,x,y',w',w\}$ induces a fork) and $v$ and $w'$ are non-adjacent (otherwise $\{u,x,z,v,w'\}$ induces a fork). Then $c$ and $w'$ are adjacent (otherwise $\{v,c,u,w,w'\}$ induces a fork) and $w'$ and $y$ are adjacent (otherwise $\{u,c,w,w',y\}$ induces a fork). Therefore, the vertex $w'$ is an $X$-neighbor of $w$.
            \item Suppose that $y' \in B_X$ slides to a neighbor $y''$. Then $v$ and $y''$ are non-adjacent (otherwise $\{u,c,w,v,y''\}$ induces a fork) and $c$ and $y''$ are non-adjacent (by Observation~\ref{lem:only-reducible}). Furthermore, $y$ and $y''$ are adjacent (otherwise $\{w,y,v,y',y''\}$ induces a fork) and $x$ and $y''$ are adjacent (otherwise $\{u,x,v,y',y''\}$ induces a fork). Therefore, the vertex $y''$ is an $X$-neighbor of $w$.
        \end{itemize}
    \end{claimproof}
    The next two claims follow from the previous ones by noticing that renaming the vertices $u$ and $w$ as well as $x$ and $y$ is an automorphism of $H_5$.
    \begin{claim}
        We can move the token $v$ of $I-u+w$ to $u$ via $x$.
    \end{claim}
    \begin{claim}
        Either we can move the token on $u$ of $I-v+w$ to $v$ via $x$ or the vertex $c$ is permanently blocked.
    \end{claim}
\end{proof}

\begin{figure}
    \centering
        \begin{tikzpicture}[vertex/.style={shape=circle,thick,draw,node distance=3em,inner sep=0.15em},edge/.style={draw,thick}]
            \node[vertex,label=above:$u$] (u1) {};
            \node[right of=u1] (dots1) {$\cdots$};
            \node[vertex,right of=dots1,label={[xshift=2mm]above left:$u_{i_{j-1}}$}] (u2) {};
            \node[vertex,right of=u2,label=above:$u_{i_j}$] (u3) {};
            \node[vertex,right of=u3,label={[xshift=-1mm]above right:$u_{i_{j+1}}$}] (u4) {};
            \node[right of=u4] (dots2) {$\cdots$};
            \node[vertex,right of=dots2,label=above:$v$] (u5) {};

            \node[vertex,above of=u2,label=above:$a_{j-1}$] (ajm1) {};
            \node[vertex,below of=u3,label=below:$a_{j}$] (aj) {};
            \node[vertex,above of=u4,label=above:$a_{j+1}$] (ajp1) {};

            \draw[edge] (u2) -- (u3) -- (u4);
            \draw[edge] (ajm1) -- (u2);
            \draw[edge] (aj) -- (u3);
            \draw[edge] (ajp1) -- (u4);

            \draw[edge,dashed] (u3) -- (ajm1) -- (u4);
            \draw [edge] (u1.0) -- +(0:0.3);
            \draw [edge] (u2.180) -- +(180:0.3);
            \draw [edge] (u4.0) -- +(0:0.3);
            \draw [edge] (u5.180) -- +(180:0.3);

            \path (ajm1.center) coordinate[draw,rectangle,minimum size=3.5mm,inner sep=0pt,outer sep=0pt,Red];
            \path (aj.center) coordinate[draw,rectangle,minimum size=3.5mm,inner sep=0pt,outer sep=0pt,Red];
            \path (ajp1.center) coordinate[draw,rectangle,minimum size=3.5mm,inner sep=0pt,outer sep=0pt,Red];
        \end{tikzpicture}
    \caption{An illustration of the situation in the proof of Theorem \ref{thm:ts_tj_equiv_fork}. The red squares represent the initial positions of the tokens. The dashed lines represent non-edges.}
    \label{fig:leftmost-neighbor}
\end{figure}
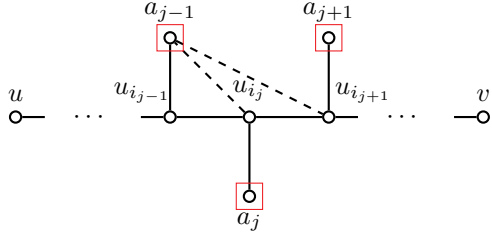

We now have all the ingredients to prove Theorem \ref{thm:ts_tj_equiv_fork}.

\begin{proof}[Proof of Theorem \ref{thm:ts_tj_equiv_fork}]
    Let $I$ be an independent set of $G$ such that Reduction Rule A is not applicable and let $u$ be a vertex of $G$ that is $I$-free. Let $v \in I$ and $J := I - v + u$. Since $u$ is $I$-free, $J$ is an independent set of $G$. We show that $I \TSeq J$. Let $P := u_1, \ldots, u_{\ell}$ be a shortest $u$-$v$ path in $G$. Since $u$ is $I$-free, we have $\ell \geq 3$. If $\ell = 3$ then either we can directly slide from $v$ to $u$ through $u_2$ or there is a token on $z \in N(u_2)$. In this case, the vertices $\{u_2, u, v, z\}$ induce a claw and by Lemma \ref{lem:rotation-claw}, we can either ``rotate'' the tokens on the claw to reach $J$, or $u = u_1$ is permanently blocked. So we may assume that $\ell \geq 4$. By Lemma \ref{lem:no-two-neigh-path}, the leftmost neighbors on $P$ of the $I$-tokens in the neighborhood of $P$ are distinct. 
    Let $\{a_1, \ldots a_t\} := N[P] \cap I$ with respective leftmost neighbors $u_{i_1}, \ldots, u_{i_t}$ on $P$, such that $i_{j} < i_{j+1}$ for all $j$. We show by induction on $i$ that there exists a sequence of slides that moves the token initially on $a_{i+1}$ to $a_i$ for $1 < i < t$  in increasing order, where $a_0 := u$.
    \begin{description}
        \item[Case i=1.] Slide the token on $a_1$ to $a_0$ via $u_{i_1}$ and then slide from $u_{i_1}$ to $u$ along $P$. This is possible since the other tokens on $N[P]$ lie on $a_2, \ldots, a_t$, which have no neighbors in $u_0, \ldots, u_{i_1}$ by the definition of leftmost neighbors.
        \item[Induction step.]  Let $1 < j < t$ and assume that the tokens initially on $a_1, \ldots, a_j$ have moved respectively to $a_0, \ldots, a_{j-1}$. We show that the token on $a_{j+1}$ can slide to $a_j$. The tokens in $N[P]$ currently are on $\{a_0, \ldots, a_t\} \setminus \{a_j\}$.\newline
        We first slide the token on $a_{j+1}$ to $u_{i_{j+1}}$. By the definition of leftmost neighbors, none of $a_{j+2}, \ldots a_t$ is adjacent to $u_{i_{j+1}}$. Furthermore, none of $a_0, \ldots, a_{j-2}$ is adjacent to $u_{i_{j+1}}$ since $P$ is a shortest $u$-$v$ path. Therefore, the only token that may prevent us from sliding $a_{j+1}$ to $u_{i_{j+1}}$ is on $a_{j-1}$, namely if the edge $a_{j-1}u_{i_{j+1}}$ is present. Assume for a contradiction that this is the case. Then $u_{i_{j-1}}$ is at distance two from $u_{i_{j+1}}$ in $P$ since $P$ is a shortest path. By Lemma~\ref{lem:no-two-neigh-path}, the leftmost neighbors of $a_1, \ldots, a_t$ are distinct, so $u_{i_{j-1}}, u_{i_j}, u_{i_{j+1}}$ are three consecutive vertices of $P$. See Figure \ref{fig:leftmost-neighbor} for an illustration. Since $u_{i_j}$ is a leftmost neighbor of $a_j$, we have that $a_j$ and $u_{i_{j-1}}$ are non-adjacent. Since $a_{j-1}$ and $a_{j+1}$ are adjacent to $u_{i_{j+1}}$ and Reduction Rule A is not applicable ($a_{j-1}, a_j, a_{j+1} \in I$), we have that $a_{j}$ is non-adjacent to $u_{i_{j+1}}$. But then $\{u_{i_{j-1}}, u_{i_j}, a_j, u_{i_{j+1}}, a_{j+1}\}$ induces a fork, so $a_{j-1}$ is non-adjacent to $u_{i_{j+1}}$. We may therefore slide the token on $a_{j+1}$ to $u_{i_{j+1}}$. 
        \newline 
        If $a_{i_{j+1}}$ is adjacent to $a_j$, we can slide the token on $a_{i_{j+1}}$ to its destination $u_j$. Otherwise, the only obstruction that prevents us from moving the token on $a_{i_{j+1}}$ to $a_{i_{j+1}-1}$ is the presence of the edge $a_{j-1}a_{i_{j+1}-1}$. But then $u_{i_j}$, $u_{j}$, and $u_{i_{j+1}-1} = u_{i_j + 1}$ are three consecutive vertices of $P$ since $P$ is a shortest path. Then $u_{i_j + 1}$ is adjacent to $a_j$, since otherwise $\{u_{i_{j-1}}, u_{i_j}, a_j, u_{i_j+1}, u_{i_{j+1}}\}$ induces a fork. But then $\{a_{j-1}, u_{i_j}, a_j, u_{i_j + 1}, u_{i_j}, a_{j+1}\}$ induces a fork. Therefore, we can slide the token on $u_{i_{j+1}} = u_{i_j+2}$ to $u_{i_j + 1}$. To conclude, observe that $u_{i_j+1}$ is adjacent to $a_j$, since otherwise $\{u_{i_{j-1}}, u_{i_j}, a_j, u_{i_j+1}, u_{i_j+1}\}$ induces a fork. Therefore, we can slide the token on $u_{i_j+2}$ to its destination $a_j$.
    \end{description}
        Therefore, $I \TSeq I - v + u$ or $G$ contains a permanently blocked vertex.
        In order to reach $I-v+u$ from $I$, we move tokens along a shortest $u$-$v$ path $P$ to their destination as described above, which can be done in polynomial time.
\end{proof}

\subsection{A polynomial-time algorithm for Token Sliding in fork-free graphs}\label{sec:resolve-cycles}

Using the reductions from the previous sections we provide a polynomial-time algorithm for \TSR in fork-free graphs. 
A \emph{complex} is a complete bipartite graph minus a matching. We recall the  following characterization of bipartite fork-free graphs due to Alekseev:

\begin{theorem}[\cite{alekseev_polynomial_2004}]\label{thm:alekseev}
    A bipartite and connected fork-free graph $G$ is either a path, a cycle or a complex.
\end{theorem}

In the following, let $G$ be a fork-free graph and $I$ and $J$ be two independent sets of $G$. Since the symmetric difference of two independent sets induces a bipartite graph, $I \Delta J$ is a collection of disjoint paths, cycles, complex and isolated vertices. The following Lemma shows that we can reduces to the case where it does not contain any complex that is not a cycle or a path:

\begin{lemma}\label{lem:sym-diff-fork-free}
    If $G$ is $I$-reduced and $J$-reduced then $I \Delta J$ is a collection of disjoint paths, cycles and isolated vertices.
\end{lemma}
\begin{proof}
    Since Reduction Rule A is not applicable, each vertex of $I$ has at most two neighbors in $J$ and vice-versa. Therefore, each connected component of  $G[I \Delta J]$ is a cycle or a path. 
\end{proof}

Let $C$ be a non-trivial connected component of $G[I \Delta J]$. Let $I' := (I \setminus V(G)) \cup (J \cap V(C))$. We say that $C$ can be \emph{resolved} if $I \TSeq I'$.
That is, $C$ can be resolved if there exists a reconfiguration sequence that replaces on $C$ the tokens of $I$ with the tokens of $J$. By abuse of notation, we say that an isolated vertex $u$ of $G[I \setminus J]$ can be resolved if there exists a reconfiguration sequence that slides $u$ to an isolated vertex of $G[J \setminus I]$. Note that if $C$ is a path it can be resolved simply by  sliding tokens one after the other along the path. If $C$ is a cycle we proceed as follows.
If there exists a free vertex $u$ with respect to the current independent set $I$, we pick an arbitrary token on $C \cap I$ and slide it to $v$ or find a permanently blocked vertex according to Theorem~\ref{thm:ts_tj_equiv_fork}. In the former case we slide the remaining tokens on $C \cap I$ to $C \cap J$ just as for paths and then slide the token on $v$ back to $C$ (again according to Theorem~\ref{thm:ts_tj_equiv_fork}). We can suppose that $I$ is not maximum since otherwise the reduction to claw-free graphs of Section \ref{sec:max-isr} applies. If $I$ is maximal, we show next that we can always ``free'' a vertex to resolve a cycle as explained before.

We first need a few definitions. Let $G$ be a graph and $A$ be an independent set of $G$. An independent set $X$ of $G - A$ is called \emph{magnifier of $A$} if the independent set $(A \cup X) \setminus (N(X) \cap A)$ is strictly larger that $A$ and the corresponding subgraph $G[X \cup (N(X) \cap A)]$ is called \emph{$A$-augmenting graph}~\cite{alekseev_polynomial_2004}. Notice that augmenting graphs are bipartite. Augmenting graphs generalize augmenting paths and are important for computing maximum independent sets. 

\begin{lemma}\label{lem:augmenting-fork-free}
    If $G$ is $I$-reduced and $I$ is maximal but not maximum then there exists an $I$-augmenting path.
\end{lemma}
\begin{proof}
    Since $I$ is not maximum, there exists an $I$-augmenting graph $H$ and since $I$ is maximal, $H$ has at least two vertices. Since $G$ is fork-free, Theorem~\ref{thm:alekseev} ensures that $H$ is either a path or a complex. Since furthermore Reduction Rule A is not applicable, we have that $H$ is an $I$-augmenting path.
\end{proof}

We use the following special case of the result of Gerber et al., who showed that $I$-augmenting paths can be computed in polynomial time in graphs with no induced skew stars. 

\begin{theorem}[\cite{GERBER2006352}]\label{thm:augmenting-chain-fork}
    There is a polynomial-time algorithm that finds an $I$-augmenting path in $G$ if it exists.
\end{theorem}
We now have everything at hand to show how to resolve cycles of the symmetric difference in polynomial time.

\begin{lemma}\label{lem:resolve-isolated-cycles}
    Assume that $G$ is prime, $I$-reduced and $J$-reduced, and that $I$ and $J$ are not maximum. Then we can find in polynomial-time either a $C$-resolving sequence or a permanently blocked vertex of $G$.
\end{lemma}
\begin{proof}
There are two cases to consider. First, suppose that $I$ is not maximal. Then there exists a vertex $u \notin N[I]$. Pick an arbitrary token on $v$ in $C$. Since $u$ is $I$-free, Theorem~\ref{thm:ts_tj_equiv_fork} ensures that we can find in polynomial time either a locally blocked vertex or a sequence of slides from $I$ to $A := I - u + v$. 
    In the later case, we can first apply the sequence from $I$ to $A$ and then move the remaining tokens on $C \cap A$ to $C \cap J$. Let $A'$ denote the obtained independent set: since Reduction Rule A is not applicable with respect to $I$ and $I \TSeq A'$, we have that Reduction Rule A is not applicable with respect to $A'$. Furthermore, there exists a vertex $u' \in C$ that is $A'$-free. We can therefore apply Theorem~\ref{thm:ts_tj_equiv_fork} to find in polynomial time either a sequence of slides from $A' - u + u'$ from $A'$, thus resolving the cycle $C$, or a locally blocked vertex of $G$.

    Now suppose that $I$ is maximal, so by Lemma~\ref{lem:augmenting-fork-free}, there must exist a $I$-augmenting path. By Theorem \ref{thm:augmenting-chain-fork}, such a path $v_1,u_1, \ldots u_k, v_k$ can be found in polynomial time, where $u_1, \ldots, u_k$ denote the vertices of the path that are in $I$. After sliding for $1 \leq i < k$ in increasing order the token on $u_i$ to $v_i$, the vertex $v_k$ has no token in its neighborhood. Therefore, the previous case applies and we may then slide the token of the chain back to their initial position 
\end{proof}

We now have everything we need to prove the main result of this section.

\begin{theorem}
    \TSR in fork-free graphs admits a polynomial-time algorithm.
\end{theorem}
\begin{proof}
    Let $(G, I, J)$ be an instance of \TSR such that $G$ is fork-free. 
    We may assume that $I$ and $J$ are not maximum, since otherwise Theorem~\ref{thm:misr-fork-me} applies. This can be tested in polynomial time by the algorithm given in~\cite{alekseev_polynomial_2004}.
    Otherwise, we invoke \cref{thm:reduction-to-prime} and obtain in polynomial time an equivalent instance $(G', I', J')$, such that $G'$ is fork-free and each connected component is prime (or $(G', I', J')$ is a trivial \NO-instance).

Since $(G', I, J)$ is a \YES-instance if and only if $(C, I \cap V(C), J \cap V(C))$ is a \YES-instance for each connected component $C$ of $G'$, we may assume that $G'$ is connected.
By Lemma \ref{lem:sym-diff-fork-free}, the graph $G'[I \Delta J]$ is a collection of disjoint paths, cycles and isolated vertices. If $G'[I \Delta J]$ contains a path $P$ we return the sequence sliding the tokens from $I \cap P$ to $J \cap P$ one by one. If $G'[I \Delta J]$ contains two isolated vertices $u \in I$ and $v \in J$, Theorem~\ref{thm:ts_tj_equiv_fork} ensures that we can either find a sequence from $I$ to $I-u+v$ or a permanently blocked vertex polynomial time. If $G'[I \Delta J]$ contains a cycle $C$, by Lemma~\ref{lem:resolve-isolated-cycles}, we can either find a resolving sequence for $C$ or a permanently blocked vertex. If at any step we find a permanently blocked vertex, we can delete it from $G'$ by \cref{lem:ruleZ} and restart the algorithm. 
Since furthermore the reduction rules A, B, D, E can be applied exhaustively in polynomial time, the overall running time of the algorithm is polynomial.
\end{proof}

\section{Conclusion}
\label{sec:conclusion}
We showed \TSR and \TJR remain \PSPACE-complete on $H$-free graphs, unless $H$ is a path, a claw, or a subdivision of the claw. Furthermore, we showed that \TSR admits a polynomial-time algorithm on fork-free graphs. The overall goal is the same as for the nominal problem, namely to obtain a complete classification of the complexity of ISR on $H$-free graphs. In this direction, an interesting open question is whether \TSR admits a polynomial-time algorithm on fork-free graphs if the independent sets are not maximum. There are two main reasons why our approach for \TSR does not work for \TJR. First, we cannot assume that tokens stay in their respective modules, so the modular decomposition seems less useful. Secondly, the irrelevant vertices for \TSR may not be irrelevant for \TJR, for example, vertices that are adjacent to three tokens. This means in particular that in order to resolve a cycle, it is not sufficient to consider augmenting chains, but it may be necessary to deal with large complexes induced by the symmetric difference of the independent sets.
A second interesting question is whether ISR admits a polynomial-time algorithm on $P_5$-free graphs. As a first step, the case of maximum independent sets could be considered.

\bibliographystyle{alpha}
\bibliography{biblio}

\newcommand{\etalchar}[1]{$^{#1}$}
\begin{thebibliography}{{\VAN{Heuvel}{}{van den}}vdH13}

\bibitem[Ale82]{alekseev1982effect}
Vladimir~E Alekseev.
\newblock The effect of local constraints on the complexity of determination of
  the graph independence number.
\newblock {\em Combinatorial-algebraic methods in applied mathematics}, pages
  3--13, 1982.

\bibitem[Ale04]{alekseev_polynomial_2004}
Vladimir~E Alekseev.
\newblock Polynomial algorithm for finding the largest independent sets in
  graphs without forks.
\newblock {\em Discrete Applied Mathematics}, 135(1-3):3--16, January 2004.

\bibitem[Bar21]{Bartier21}
Valentin Bartier.
\newblock {\em Combinatorial and Algorithmic aspects of Reconfiguration.
  (Aspects combinatoires et algorithmiques de la Reconfiguration)}.
\newblock PhD thesis, Grenoble Alpes University, France, 2021.

\bibitem[BB17]{bodlaender_token_2017}
Marthe Bonamy and Nicolas Bousquet.
\newblock Token {Sliding} on {Chordal} {Graphs}.
\newblock In Hans~L. Bodlaender and Gerhard~J. Woeginger, editors, {\em
  Graph-{Theoretic} {Concepts} in {Computer} {Science}}, volume 10520, pages
  127--139. Springer International Publishing, Cham, 2017.
\newblock Series Title: Lecture Notes in Computer Science.

\bibitem[BBD{\etalchar{+}}21]{Bartier2020OnGA}
Valentin Bartier, Nicolas Bousquet, Cl{\'{e}}ment Dallard, Kyle Lomer, and
  Amer~E. Mouawad.
\newblock On girth and the parameterized complexity of token sliding and token
  jumping.
\newblock {\em Algorithmica}, 83(9):2914--2951, 2021.

\bibitem[BCM23]{belavadi23}
Manoj Belavadi, Kathie Cameron, and Owen Merkel.
\newblock Reconfiguration of vertex colouring and forbidden induced subgraphs,
  2023.

\bibitem[BDK{\etalchar{+}}23]{BDKLJ:23}
Marthe Bonamy, Oscar Defrain, Tereza Klimosov{\'{a}}, Aur{\'{e}}lie Lagoutte,
  and Jonathan Narboni.
\newblock On {V}izing's edge colouring question.
\newblock {\em J. Comb. Theory, Ser. {B}}, 159:126--139, 2023.

\bibitem[BDO21]{BONAMY20216}
Marthe Bonamy, Paul Dorbec, and Paul Ouvrard.
\newblock Dominating sets reconfiguration under token sliding.
\newblock {\em Discret. Appl. Math.}, 301:6--18, 2021.

\bibitem[BHI{\etalchar{+}}20]{BONAMY202045}
Marthe Bonamy, Marc Heinrich, Takehiro Ito, Yusuke Kobayashi, Haruka Mizuta,
  Moritz M{\"{u}}hlenthaler, Akira Suzuki, and Kunihiro Wasa.
\newblock Diameter of colorings under {K}empe changes.
\newblock {\em Theor. Comput. Sci.}, 838:45--57, 2020.

\bibitem[BHV04]{brandstadt2004minimal}
Andreas Brandst{\"{a}}dt, Ch{\'{\i}}nh~T. Ho{\`{a}}ng, and Jean{-}Marie
  Vanherpe.
\newblock On minimal prime extensions of a four-vertex graph in a prime graph.
\newblock {\em Discret. Math.}, 288(1-3):9--17, 2004.

\bibitem[BKL{\etalchar{+}}21]{belmonte_token_2021}
R{\'{e}}my Belmonte, Eun~Jung Kim, Michael Lampis, Valia Mitsou, Yota Otachi,
  and Florian Sikora.
\newblock Token sliding on split graphs.
\newblock {\em Theory Comput. Syst.}, 65(4):662--686, 2021.

\bibitem[BKW14]{hutchison_reconfiguring_2014}
Paul~S. Bonsma, Marcin Kamiński, and Marcin Wrochna.
\newblock Reconfiguring independent sets in claw-free graphs.
\newblock In {\em Algorithm Theory - {SWAT} 2014 - 14th Scandinavian Symposium
  and Workshops}, volume 8503 of {\em Lecture Notes in Computer Science}, pages
  86--97. Springer, 2014.

\bibitem[BMNS22]{ISRsurvey}
Nicolas Bousquet, Amer~E. Mouawad, Naomi Nishimura, and Sebastian Siebertz.
\newblock A survey on the parameterized complexity of the independent set and
  (connected) dominating set reconfiguration problems.
\newblock {\em CoRR}, abs/2204.10526, 2022.

\bibitem[Bon16]{bonsma_independent_2016}
Paul~S. Bonsma.
\newblock Independent set reconfiguration in cographs and their
  generalizations.
\newblock {\em J. Graph Theory}, 83(2):164--195, 2016.

\bibitem[BP16]{BOUSQUET20161}
Nicolas Bousquet and Guillem Perarnau.
\newblock Fast recoloring of sparse graphs.
\newblock {\em Eur. J. Comb.}, 52:1--11, 2016.

\bibitem[DDF{\etalchar{+}}15]{demaine_linear-time_2015}
Erik~D. Demaine, Martin~L. Demaine, Eli Fox{-}Epstein, Duc~A. Hoang, Takehiro
  Ito, Hirotaka Ono, Yota Otachi, Ryuhei Uehara, and Takeshi Yamada.
\newblock Linear-time algorithm for sliding tokens on trees.
\newblock {\em Theor. Comput. Sci.}, 600:132--142, 2015.

\bibitem[DF21]{DVORAK2021103319}
Zdenek Dvor{\'{a}}k and Carl Feghali.
\newblock A {T}homassen-type method for planar graph recoloring.
\newblock {\em Eur. J. Comb.}, 95:103319, 2021.

\bibitem[FHOU15]{elbassioni_sliding_2015}
Eli Fox{-}Epstein, Duc~A. Hoang, Yota Otachi, and Ryuhei Uehara.
\newblock Sliding token on bipartite permutation graphs.
\newblock In {\em Algorithms and Computation - 26th International Symposium,
  {ISAAC} 2015, Nagoya, Japan, December 9-11, 2015, Proceedings}, pages
  237--247, 2015.

\bibitem[Gal67]{gallai1967transitiv}
Tibor Gallai.
\newblock Transitiv orientierbare {G}raphen.
\newblock {\em Acta Mathematica Hungarica}, 18(1-2):25--66, 1967.

\bibitem[GHL06]{GERBER2006352}
Michael~U. Gerber, Alain Hertz, and Vadim~V. Lozin.
\newblock Augmenting chains in graphs without a skew star.
\newblock {\em J. Comb. Theory, Ser. {B}}, 96(3):352--366, 2006.

\bibitem[GKPP22]{grzesik_polynomial-time_2020}
Andrzej Grzesik, Tereza Klimosov{\'{a}}, Marcin Pilipczuk, and Michal
  Pilipczuk.
\newblock Polynomial-time algorithm for maximum weight independent set on
  $p_{6}$-free graphs.
\newblock {\em {ACM} Trans. Algorithms}, 18(1):4:1--4:57, 2022.

\bibitem[HD05]{hearn_pspace-completeness_2005}
Robert~A. Hearn and Erik~D. Demaine.
\newblock {PSPACE}-completeness of sliding-block puzzles and other problems
  through the nondeterministic constraint logic model of computation.
\newblock {\em Theor. Comput. Sci.}, 343(1-2):72--96, 2005.

\bibitem[{\VAN{Heuvel}{}{van den}}vdH13]{Heuvel:13}
Jan {\VAN{Heuvel}{}{van den}}~van~den Heuvel.
\newblock The complexity of change.
\newblock In Simon~R. Blackburn, Stefanie Gerke, and Mark Wildon, editors, {\em
  Surveys in Combinatorics 2013}, volume 409 of {\em London Mathematical
  Society Lecture Note Series}, pages 127--160. Cambridge University Press,
  2013.

\bibitem[HIM{\etalchar{+}}16]{HADDADAN201637}
Arash Haddadan, Takehiro Ito, Amer~E. Mouawad, Naomi Nishimura, Hirotaka Ono,
  Akira Suzuki, and Youcef Tebbal.
\newblock The complexity of dominating set reconfiguration.
\newblock {\em Theor. Comput. Sci.}, 651:37--49, 2016.

\bibitem[IDH{\etalchar{+}}11]{ito_complexity_2011}
Takehiro Ito, Erik~D. Demaine, Nicholas J.~A. Harvey, Christos~H.
  Papadimitriou, Martha Sideri, Ryuhei Uehara, and Yushi Uno.
\newblock On the complexity of reconfiguration problems.
\newblock {\em Theor. Comput. Sci.}, 412(12-14):1054--1065, 2011.

\bibitem[IKO14a]{ahn_fixed-parameter_2014}
Takehiro Ito, Marcin Kamiński, and Hirotaka Ono.
\newblock Fixed-parameter tractability of token jumping on planar graphs.
\newblock In {\em Algorithms and Computation - 25th International Symposium,
  {ISAAC} 2014}, volume 8889 of {\em Lecture Notes in Computer Science}, pages
  208--219. Springer, 2014.

\bibitem[IKO{\etalchar{+}}14b]{hutchison_parameterized_2014}
Takehiro Ito, Marcin Kamiński, Hirotaka Ono, Akira Suzuki, Ryuhei Uehara, and
  Katsuhisa Yamanaka.
\newblock On the parameterized complexity for token jumping on graphs.
\newblock In {\em Theory and Applications of Models of Computation - 11th
  Annual Conference, {TAMC}. Proceedings}, volume 8402 of {\em Lecture Notes in
  Computer Science}, pages 341--351. Springer, 2014.

\bibitem[IKO{\etalchar{+}}20]{Ito:20}
Takehiro Ito, Marcin Kamiński, Hirotaka Ono, Akira Suzuki, Ryuhei Uehara, and
  Katsuhisa Yamanaka.
\newblock Parameterized complexity of independent set reconfiguration problems.
\newblock {\em Discret. Appl. Math.}, 283:336--345, 2020.

\bibitem[KMM12]{kaminski_complexity_2012}
Marcin Kamiński, Paul Medvedev, and Martin Milanič.
\newblock Complexity of independent set reconfigurability problems.
\newblock {\em Theor. Comput. Sci.}, 439:9--15, 2012.

\bibitem[LM08]{lozin_polynomial_2008}
Vadim~V. Lozin and Martin Milanič.
\newblock A polynomial algorithm to find an independent set of maximum weight
  in a fork-free graph.
\newblock {\em Journal of Discrete Algorithms}, 6(4):595--604, December 2008.

\bibitem[LM19]{lokshtanov_complexity_2019}
Daniel Lokshtanov and Amer~E. Mouawad.
\newblock The {Complexity} of {Independent} {Set} {Reconfiguration} on
  {Bipartite} {Graphs}.
\newblock {\em ACM Transactions on Algorithms}, 15(1):1--19, January 2019.

\bibitem[LVV14]{lokshantov_independent_2014}
Daniel Lokshtanov, Martin Vatshelle, and Yngve Villanger.
\newblock Independent set in ${P}_5$-free graphs in polynomial time.
\newblock In {\em Proceedings of the Twenty-Fifth Annual {ACM-SIAM} Symposium
  on Discrete Algorithms, {SODA} 2014}, pages 570--581. {SIAM}, 2014.

\bibitem[Min80]{minty_maximal_1980}
George~J. Minty.
\newblock On maximal independent sets of vertices in claw-free graphs.
\newblock {\em J. Comb. Theory, Ser. {B}}, 28(3):284--304, 1980.

\bibitem[MNR{\etalchar{+}}17]{hutchison_parameterized_2013}
Amer~E. Mouawad, Naomi Nishimura, Venkatesh Raman, Narges Simjour, and Akira
  Suzuki.
\newblock On the parameterized complexity of reconfiguration problems.
\newblock {\em Algorithmica}, 78(1):274--297, 2017.

\bibitem[MNRS18]{mouawadVC2014}
Amer~E. Mouawad, Naomi Nishimura, Venkatesh Raman, and Sebastian Siebertz.
\newblock Vertex cover reconfiguration and beyond.
\newblock {\em Algorithms}, 11(2):20, 2018.

\bibitem[M{\"o}h85]{mohring1985algorithmic}
Rolf~H. M{\"o}hring.
\newblock Algorithmic aspects of the substitution decomposition in optimization
  over relations, set systems and boolean functions.
\newblock {\em Annals of Operations Research}, 4:195--225, 1985.

\bibitem[MR84]{mohring1984substitution}
Rolf~H. M{\"o}hring and Franz~J. Radermacher.
\newblock Substitution decomposition for discrete structures and connections
  with combinatorial optimization.
\newblock In {\em North-Holland mathematics studies}, volume~95, pages
  257--355. Elsevier, 1984.

\bibitem[MS99]{mcconnell1999modular}
Ross~M. McConnell and Jeremy~P. Spinrad.
\newblock Modular decomposition and transitive orientation.
\newblock {\em Discrete Mathematics}, 201(1-3):189--241, 1999.

\bibitem[Nis18]{Nishimura:18}
Naomi Nishimura.
\newblock Introduction to reconfiguration.
\newblock {\em Algorithms}, 11(4):52, 2018.

\bibitem[Sbi80]{sbihi_algorithme_1980}
Najiba Sbihi.
\newblock Algorithme de recherche d'un stable de cardinalite maximum dans un
  graphe sans etoile.
\newblock {\em Discrete Mathematics}, 29(1):53--76, 1980.

\bibitem[SMN16]{Suzuki2016}
Akira Suzuki, Amer~E. Mouawad, and Naomi Nishimura.
\newblock Reconfiguration of dominating sets.
\newblock {\em Journal of Combinatorial Optimization}, 32(4):1182--1195, Nov
  2016.

\end{thebibliography}

\end{document}